\RequirePackage{fix-cm}
\documentclass[smallextended]{svjour3}       % onecolumn (second format)
\smartqed  % flush right qed marks, e.g. at end of proof
\usepackage{graphicx,color}
\usepackage{relsize}
\usepackage{url}
\usepackage{amsmath}
\usepackage{amssymb}
\usepackage{amsfonts}
\usepackage{algorithm}
\usepackage{algpseudocode}

\usepackage{bm}

\newcommand{\e}{\text{e}}
\newcommand{\tr}{\text{tr}}
\newcommand{\Var}{\text{Var}}

\newcommand{\ignore}[1]{}

\DeclareMathOperator*{\argmin}{arg\,min}
\begin{document}

\title{Stochastic diagonal estimation: probabilistic bounds and an improved algorithm%\thanks{Grants or other notes
}

\titlerunning{Stochastic diagonal estimation}        % if too long for running head

\author{
Robert A. Baston 
 \and
Yuji Nakatsukasa%etc.
}


\institute{
R. A. Baston and Y. Nakatsukasa
 \at
	Mathematical Institute\\
	University of Oxford\\
              Tel.: +44-1865-615319\\
%              \email{fauthor@example.com}           %  \\
	\email{robert.baston@maths.ox.ac.uk}, 
	\email{nakatsukasa@maths.ox.ac.uk} 
%           \and           S. Author \at              second address
}

\date{Received: date / Accepted: date}

\maketitle

\begin{abstract}
We study the problem of estimating the diagonal of an implicitly given matrix $A$. For such a matrix we have access to an oracle that allows us to evaluate the matrix vector product $A\bm{v}$. For random variable $\bm{v}$ drawn from an appropriate distribution, this may be used to return an estimate of the diagonal of the matrix $A$. 
Whilst results exist for probabilistic guarantees relating to the error of estimates of the trace of $A$, no such results have yet been derived for the diagonal. We make two contributions in this regard. We analyse the number of queries $s$ required to guarantee that with probability at least $1-\delta$ the estimates of the relative error of the diagonal entries is at most $\varepsilon$. We extend this analysis to the 2-norm of the difference between the estimate and the diagonal of $A$. We prove, discuss and experiment with bounds on the number of queries $s$ required to guarantee a probabilistic bound on the estimates of the diagonal by employing Rademacher and Gaussian random variables. Two sufficient upper bounds on the minimum number of query vectors are proved, extending the work of Avron and Toledo 
[JACM 58(2)8, 2011], and later work of Roosta-Khorasani and Ascher [FoCM 15, 1187-1212, 2015]. We first prove sufficient bounds for diagonal estimation with Gaussian vectors, and then similarly with Rademacher vectors. We find that, generally, there is little difference between the two, with convergence going as $O(\log(1/\delta)/\varepsilon^2)$ for individual diagonal elements. However for small $s$, we find that the Rademacher estimator is superior. These results allow us to then extend the ideas of Meyer, Musco, Musco and Woodruff [SOSA, 142-155, 2021], suggesting algorithm Diag++, to speed up the convergence of diagonal estimation from $O(1/\varepsilon^2)$ to $O(1/\varepsilon)$ and make it robust to the spectrum of any positive semi-definite matrix $A$. We analyse our algorithm to demonstrate state-of-the-art convergence.
%Insert your abstract here. Include keywords, PACS and mathematical subject classification numbers as needed.
%\keywords{Trace estimation \and Second keyword \and More}
\end{abstract}

\section{Introduction} \label{sec:Introduction}
An ever-present problem in numerical linear algebra relates to estimating the trace of a matrix $A\in\mathbb{R}^{n\times n}$ that may only be accessed through its operation on a given vector. Such a calculation is required in a wide variety of applications, including: matrix norm estimation \cite{han2017approximating,musco2017spectrum}, image restoration \cite{golub1979generalized,golub1997generalized}, density functional theory \cite{bekas2007estimator} and log-determinant approximation \cite{boutsidis2017randomized,han2015large}. The concepts surrounding diagonal estimation are younger than those of trace estimation, but have recently come of age in multiple applications. The theory however, is lagging. Originally used for electronic structure calculations  \cite{bekas2007estimator,goedecker1999linear,goedecker1995tight}, the ideas developed for diagonal estimation have recently been used by Yao et.\ al 
\cite{yao2020adahessian} in a novel, state-of-the-art machine learning algorithm. They estimate the diagonal of a Hessian matrix to ``precondition" descent vectors, for Newton-style descent, in the training of convolution neural networks. Their method produces exceptional accuracy across a wide array of problems. By understanding the properties of diagonal estimation more fully we may uncover new algorithms or methods to improve even these results.

 Such estimations involve a matrix $A$ which is a ``black box". We are given an oracle that may evaluate $A\bm{v}$ for any $\bm{v} \in \mathbb{R}^n$, and we seek to return the chosen property of $A$ using as few query vectors $\bm{v}$ as possible: with $s$ queries, we desire $s \ll n$. The first exploration into this area was by Hutchinson \cite{hutchinson1989stochastic}, which resulted in the standard approach to estimate the trace of a matrix:
\begin{equation}
    \tr(A) \approx \tr_R^s(A):= \frac{1}{s}\sum_{k=1}^s \bm{v}_k^T A\bm{v}_k.
\end{equation}
The $\bm{v}_k \in \mathbb{R}^n$ are vectors with entries given as independent, identically distributed Rademacher variables: the $i^{th}$ entry of $\bm{v}_k$ takes the value $-1$ or $1$ with equal probability for all $i$, and independence across both $i$ and $k$. The subscript $R$ signifies the distribution of the entries of $\bm{v}_k$ to be Rademacher.

 Later work expanded on this idea, using Gaussian vectors distributed as $\bm{v}_k \sim N(0, I_n)$, but until the paper of Avron and Toledo \cite{avron2011randomized}, analysis and comparison of these estimators was limited to explorations of their variance. It is well known \cite{hutchinson1989stochastic} that the variance of the Hutchinson method is the lowest when compared to other standard methods, thus is used most extensively in applications. However, whilst intuition suggests this is a good choice, it is not completely rigorous. The techniques of \cite{avron2011randomized} address the error bounds in probability directly.
\subsection*{The $(\varepsilon, \delta)$ approach}
\noindent In \cite{avron2011randomized} the authors propose so-called $(\varepsilon,\delta)$ bounds. A lower bound on the number of queries, $s$, sufficient to achieve a relative error of the estimated trace, is obtained. This is a probabilistic guarantee. Formally, for a pair of $(\varepsilon,\delta)$ values, a sufficient number of queries, $s$, is found, such that, for positive semi-definite matrix $A$
\begin{equation} \label{Trace Estimator Intro}
   \Pr\Big(|\tr_{\mathbb{P}}^s(A) - \tr(A)|\leq\varepsilon\tr(A)\Big)\geq 1-\delta.
\end{equation}
The term $\tr_{\mathbb{P}}^s(A)$ is the $s$-query estimator of the trace, using query vectors drawn from the distribution $\mathbb{P}$. Supposing $s$ random vectors have been used to query the matrix, we have, with high probability, a relative error that is at most $\varepsilon$, where $\varepsilon$ is dependent on $s$. The authors conclude that analysis of the variance of the estimators is not the most informative approach for selecting which estimator is likely to return the best results for the fewest queries. This is a key insight.\\

\subsection*{The diagonal estimator.}

\noindent Bekas et.\ al \cite{bekas2007estimator} extended Hutchinson's original method to diagonal estimation by considering the Hadamard product of $\bm{v}_k$ and $A\bm{v}_k$. The method is different from trace estimation in two regards. Firstly, in trace estimation, we sum the element-wise multiplication terms to get the inner product of $\bm{v}_k$ and $A\bm{v}_k$. This step is removed for diagonal estimation. Secondly, the estimation is not necessarily scaled by the number of queries, but by the element-size of the query vectors. This gives
\begin{equation} \label{Diagonal Estimator Intro}
    \bm{D}^s = \Big[ \sum_{k=1}^s \bm{v}_k \odot A \bm{v}_k \Big] \oslash \Big[\sum_{k=1}^s \bm{v}_k \odot \bm{v}_k \Big].
\end{equation}
The method returns $\bm{D}^s$, an estimation of diag($A$), the diagonal of $A$ reshaped as an $n$-dimensional vector. The symbol $\odot$ represents component-wise multiplication of vectors and $\oslash$ represents component-wise division. In \cite{bekas2007estimator}, the authors focus on employing deterministic vectors for $\bm{v}_k$, no further treatment of the stochastic estimators is given, except in numerical experiments.

\subsection*{Contributions}
\noindent In this paper, we proceed to consider the same aims as those of Avron and Toledo \cite{avron2011randomized} and those of Bekas et.\ al \cite{bekas2007estimator}: applying the ideas of \cite{avron2011randomized} to analyse the stochastic diagonal estimator suggested in \cite{bekas2007estimator} and shown in equation \eqref{Diagonal Estimator Intro}. No paper to date has presented a theoretical bound on the number of queries $s$ required to achieve an $\varepsilon$ approximation of the diagonal of a matrix. Specifically, we prove upper bounds on the minimum number of query vectors required for the relative (and ``row-dependent") error of the entries of the diagonal estimate to be at most $\varepsilon$, with probability $1-\delta$. We extend this result to consider the error across all diagonal entries, and the associated number of queries required for a similar error bound. We consider estimators whose query vectors $\bm{v}_k$ obey different distributions. This analysis allows us to then suggest and analyse a novel algorithm for speed-up of diagonal estimation. Extending the work in \cite{meyer2021hutch++} for trace estimation, we propose an algorithm for diagonal estimation that shifts from the classic Monte-Carlo regime of $O(1/\varepsilon^2)$ to $O(1/\varepsilon)$, a fundamental, state-of-the-art improvement for diagonal estimation.
Much of the paper is based on the first author's MSc dissertation~\cite{baston2021thesis}. 

\subsection*{Outline of this work}
This paper is outlined as follows. In Section \ref{sec:Preliminaries} we outline our use of notation, provide a brief overview of previous work, and outline the relation of the Rademacher diagonal estimator to Hutchinson's trace estimator. In Section \ref{sec: Definitions and Estimators} we define what is meant by an $(\varepsilon, \delta)$-approximator in the context of diagonal estimation, and outline two types of stochastic diagonal estimators, Gaussian and Rademacher. In Section \ref{Comparing the Quality of estimators.} we state our results, proved in the following sections. We prove an upper bound on the minimum number of sufficient queries for an $(\varepsilon,\delta)$-approximator using the Gaussian diagonal estimator in Section \ref{Gaussian section}, and then for the Rademacher diagonal estimator in Section \ref{Rademacher section}. Section \ref{Positive-semi-definite diagonal estimation} manipulates and explores these results in an attempt to make them more useful to the practitioner. We perform numerical experiments illustrating the results in Section \ref{Numerical Experiments on the diagonal}. Where appropriate a touch of theory is also covered in this section. In Section \ref{Improved diagonal estimation} we use these previous results to adapt the Hutch++ algorithm from \cite{meyer2021hutch++} to suggest and analyse a similar algorithm, Diag++, that enables $O(1/\varepsilon^2)$ to $O(1/\varepsilon)$ speed-up of stochastic diagonal estimation. Section \ref{Conclusions} provides concluding remarks.

\section{Preliminaries} \label{sec:Preliminaries}
In this section we outline our use of notation, and provide a brief overview of previous work.

\subsection{Notation} \label{Notation}
For all $\bm{u}\in\mathbb{R}^n$,
$\|\bm{u}\|_2 = (\sum_{i=1}^n u_i^2)^{1/2}$
denotes the $\ell_2$ norm and likewise
$\|\bm{u}\|_1 = \sum_{i=1}^n |u_i|$
denotes the $\ell_1$ norm. For a general matrix $A\in\mathbb{R}^{m\times n}$,
$\|A\|_F = (\sum_{i=1}^m\sum_{j=1}^n A_{ij}^2)^{1/2}$
denotes the Frobenius norm of a matrix, where, of course, $A_{ij}$ is the element of $A$ from the $i^{th}$ row and the $j^{th}$ column. In the case of $A\in\mathbb{R}^{n\times n}$, 
$\tr(A) = \sum_{i=1}^n A_{ii} = \sum_{i=1}^n \lambda_{i}$
is trace of the matrix $A$, with $\lambda_i$ its eigenvalues. Sometimes in this paper we consider the matrix $A$ to be symmetric: $A=A^T$, or symmetric positive semi-definite (PSD), indicated by $A\succeq 0$. Such $A$ has an eigendecomposition of the form $A = V\Lambda V^T$
where $V\in \mathbb{R}^{n\times n}$ is the orthogonal matrix of eigenvectors of $A$ and $\Lambda$ is a real-valued diagonal matrix. We write 
$\bm{\lambda} = \text{diag}(\Lambda)$
to be a vector of length $n$ containing the eigenvalues of $A$ in descending order
$\lambda_1 \geq \lambda_2 \geq \hdots \geq \lambda_n.$
Where the differences between the eigenvalues are relatively small, we say a matrix has a ``flatter" spectrum. Where the differences between the eigenvalues are larger, we describe the spectrum as ``steep"\footnote{Whilst there is much more nuance to be found in the spectra of matrices, this description serves our purposes sufficiently, without over-complicating the details.}.
For $A\succeq 0$ we have that $\lambda_i \geq 0$ for all $i$. Where $A$ is diagonalisable we use the identity 
$\|\bm{\lambda}\|_2 \equiv \|A\|_F$
and where $A \succeq 0$ we write
$\|\bm{\lambda}\|_1 \equiv \tr(A)$
We denote
$A_r = \argmin_{B, \text{ rank}(B)=r}\|A-B\|_F$
as the optimal $r$-rank approximation to $A$. When the matrix $A$ is PSD, the eigenvalues of the matrix coincide with its singular values. Thus, by the Eckart-Young-Mirsky theorem \cite{eckart1936approximation}, for such $A$,
$A_r = V_r \Lambda_r V_r^T$
with $V_r \in \mathbb{R}^{n\times r}$ containing the first $r$ columns of $V$ and $\Lambda_r \in \mathbb{R}^{r\times r}$ being the top left sub-matrix of $\Lambda$. Finally we denote
$\bm{A}_d = \text{diag}(A)$
as reshaping of the diagonal elements of the matrix $A$ to a vector of length $n$. For this vector, we clearly have 
$\|\bm{A}_d\|_2 = (\sum_{i=1}^n A_{ii}^2)^{1/2}$
and in the case of $A\succeq 0$,
$\|\bm{A}_d\|_1 \equiv \tr(A)$, since $A_{ii}$ are necessarily non-negative for such $A$.

\subsection{Hutchinson's Method and Previous Work} \label{Hutchinson's Method and Related Work}
The classic Monte Carlo method for approximating the trace of an implicit matrix is due to Hutchinson \cite{hutchinson1989stochastic}, who states and proves the following lemma.
\begin{lemma} \label{Hutchinson's Original Lemma}
    Let $A$ be an $n\times n$ symmetric matrix. Let $\bm{v}$ be a random vector whose queries are i.i.d Rademacher random variables $(\Pr(v^i = \pm 1) = 1/2)$, then $\bm{v}^T A\bm{v}$ is an unbiased estimator of $\textnormal{tr}(A)$, that is
    \begin{equation} \label{Hutchinson's original expectation}
        \mathbb{E}\Big[\bm{v}^T A\bm{v}\Big] = \textnormal{tr}(A)
    \end{equation}
    and
    \begin{equation} \label{Hutchinson's original variance}
        \textnormal{Var}\Big(\bm{v}^T A\bm{v}\Big) = 2\Big(\|A\|_F^2 - \sum_{i=1}^n A_{ii}^2\Big) = 2\Big(\|A\|_F^2 - \|\bm{A}_d\|_2^2\Big). 
    \end{equation}
\end{lemma}
\noindent Considering the variance term, we can easily see that it captures how much of the matrix's ``energy" (i.e. the Frobenius norm) is located in the off-diagonal elements. 
Unfortunately, in the general case, this variance may be arbitrarily large. More vexatious, is the fact that Lemma \ref{Hutchinson's Original Lemma} does not provide the user with a rigorous bound on the number of matrix-vector queries for a given accuracy. To address this issue, Avron and Toledo \cite{avron2011randomized} introduced the following definition for $A \succeq 0$ (Defintion 4.1 of \cite{avron2011randomized}):
\begin{definition}
    Let $A$ be a symmetric positive semi-definite matrix. A randomised trace estimator: $\textnormal{tr}_{\mathbb{P}}^s(A)$, is an \textnormal{$(\varepsilon, \delta)$-approximator} of $\textnormal{tr}(A)$ if
    \begin{equation}
       \Pr\Big(|\textnormal{tr}_{\mathbb{P}}^s(A) - \textnormal{tr}(A)| \leq \varepsilon\textnormal{tr}(A)\Big) \geq 1 - \delta.
    \end{equation}
\end{definition}

\noindent This definition provides a better way to be able to bound the requisite number of queries. It guarantees that the probability of the relative error exceeding $\varepsilon$ is at most $\delta$.
For example, when estimating the trace according to Hutchinson's method, it may be shown \cite{roosta2015improved} that it is sufficient to take the number of queries, $s$ as 
\begin{equation} \label{Roosta Hutchinson Bound}
    s > \frac{6\ln(2/\delta)}{\varepsilon^2}
\end{equation}
for any PSD matrix $A$ (see Theorem 1, Section 2, \cite{roosta2015improved}), whereas if Gaussian vectors are used, the bound on $s$ becomes
\begin{equation} \label{Roosta Gaussian Bound}
    s > \frac{8\ln(2/\delta)}{\varepsilon^2}
\end{equation}
\noindent (Theorem 3, Section 3, \cite{roosta2015improved}). 
We note that the recent paper~\cite{cortinovis2021randomized} gives query bounds for indefinite matrices.

Both equations \eqref{Roosta Hutchinson Bound} and \eqref{Roosta Gaussian Bound} immediately put us in the Monte Carlo regime, with error converging as $O\big(1/\sqrt{s}\big)$ for $s$ queries. Let us step back and remark on the fact that the trace of a matrix is invariant under orthogonal similarity transforms of the matrix. As such, we can intuit that the loose, matrix-independent bounds in equations \eqref{Roosta Hutchinson Bound} and \eqref{Roosta Gaussian Bound} are naturally to be expected. Can we say the same of estimating the diagonal? We cannot be sure that the previous intuition applies, as the diagonal of a matrix is not (generally) invariant under orthogonal similarity transformations.

\subsection{Approximating the Diagonal of a Matrix} \label{Approximating the diagonal of a matrix}

\noindent The ideas of Hutchinson are readily extended to approximate the full diagonal, as was first introduced by Bekas et.\ al \cite{bekas2007estimator}, and which we outline here. Take a sequence of $s$ vectors $\bm{v}_1, ..., \bm{v}_s$, and suppose that their entries follow some suitable distribution, being i.i.d with expectation\footnote{Because we scale the estimation as in equation \eqref{First diagonal estimator} it may be shown that the distribution may have arbitrary variance, though $\mathbb{E}[X^2] = 1$ is used for simplicity throughout.} zero. Then the diagonal approximation of matrix $A$, expressed as an $n$-dimensional vector $\bm{D}^s$, may be taken as
\begin{equation}\label{First diagonal estimator}
    \bm{D}^s = \Bigg[ \sum_{k=1}^s \bm{v}_k \odot A \bm{v}_k \Bigg] \oslash \Bigg[\sum_{k=1}^s \bm{v}_k \odot \bm{v}_k \Bigg]
\end{equation}
where $\odot$ is the component-wise multiplication of the vector (a Hadamard product), and similarly, $\oslash$ the component-wise division of the vectors. Let us examine a single component of $\bm{D}^s$ to gain more insight. Let $v_k^i$ denote the $i^{th}$ element of the $k^{th}$ vector, we have
\begin{equation*} 
    \begin{split}
        D_i^s &= \Bigg[\sum_{k=1}^s v_k^i \sum_{j=1}^n A_{i j} v_k^j\Bigg]\Bigg/\Bigg[\sum_{k=1}^s{(v_k^i)^2}\Bigg]\\
        &= A_{ii} +\Bigg[\sum_{k=1}^s v_k^i \sum_{j \neq i}^n A_{i j} v_k^j\Bigg]\Bigg/\Bigg[\sum_{k=1}^s{(v_k^i)^2}\Bigg]\\
        &= A_{ii} + \mathlarger{\mathlarger{\sum_{j \neq i}^n}} A_{i j} \Bigg(\frac{\sum_{k=1}^s v_k^i v_k^j}{\sum_{k=1}^s{(v_k^i)^2}}\Bigg).
    \end{split}
\end{equation*}
In expectation, the coefficients of the $A_{ij}$ entries shall be zero. This is true provided that the components of the $\bm{v}_k$ vectors, $v_k^i$,  themselves have an expectation of zero, regardless of the form of the denominator\footnote{Easily shown by the chain rule of expectation.}.

\subsection{Relation to Hutchinson's method}
If the vectors $\bm{v}_k$ have i.i.d Rademacher entries, then it is clear to see that the sum of the entries of the vector $\bm{D}^s$ is nothing but the Hutchinson trace estimate. %\bb{YN: this seems rather obvious to me. I'm ok with keepint it here but a referee might ask us to remove or move to appendix.}
To show this rigorously we examine the expectation and variance of a single component $D_i^s$ of $\bm{D}^s$, when $v_k^i$ are independent Rademacher variables. We note that this is straightforward to find, and so the results are included only for completeness.
\begin{equation*}
\centering
    \begin{split}
        \mathbb{E}[D_i^s] &= \mathbb{E}\Bigg[A_{ii} + \mathlarger{\sum_{j \neq i}^n} A_{ij} \Bigg(\frac{\sum_{k=1}^s v_k^i v_k^j}{\sum_{k=1}^s{(v_k^i)^2}}\Bigg)\Bigg] \\
        &= A_{ii} + \mathlarger{\sum_{j \neq i}^n} A_{ij} \mathbb{E}\Bigg[\frac{\sum_{k=1}^s v_k^i v_k^j}{\sum_{k=1}^s{(v_k^i)^2}}\Bigg] \hspace{1.25cm} \text{(by linearity)} \\
        &= A_{ii} + \mathlarger{\sum_{j \neq i}^n} A_{ij} \mathbb{E}\Bigg[\frac{\sum_{k=1}^s v_k^i v_k^j}{s}\Bigg]\hspace{1.32cm}\text{(since $(v_k^i)^2 = 1\; \forall \;i,k$)}\\
        &= A_{ii} + \mathlarger{\sum_{j \neq i}^n} A_{ij} \frac{\sum_{k=1}^s \mathbb{E}[v_k^i]\mathbb{E}[v_k^j]}{s} \hspace{1.15cm} \text{(by independence, $j\neq i$)} \\
        &= A_{ii}
    \end{split}
\end{equation*}
and by the linearity of expectation it is easy to see that
\begin{equation}
    \mathbb{E}\Big[\sum_{i=1}^n D^s_i\Big] = \sum_{i=1}^n \mathbb{E}\Big[D^s_i\Big]= \tr(A).
\end{equation}
Specifically for a single query $s=1$ we are returned equation \eqref{Hutchinson's original expectation}, the first result of Lemma \ref{Hutchinson's Original Lemma}. We now find the variance; most easily done by considering a single estimate, $s=1$. Thus
\begin{equation} \label{Single element variance Rademacher}
    \begin{split}
        \Var(D_i^{1}) &= \Var(D_i^{1} - A_{ii})\\
        &= \mathbb{E}\Bigg[\Big(\sum_{j\neq i}^n A_{ij}v^i v^j\Big)^2\Bigg] - \mathbb{E}\Bigg[\sum_{j\neq i}^n A_{ij}v^i v^j\Bigg]^2\\
        & = \sum_{j\neq i}^n\sum_{l\neq i}^n A_{ij}A_{il}\mathbb{E}\Big[(v^i)^2 v^j v^l\Big]\\
        &= \sum_{j\neq i}^n A_{ij}^2\\
        &= \Big(\|\bm{A}_i\|_2^2 - A_{ii}^2\Big)
    \end{split}
\end{equation}
where $\bm{A}_i$ is the $i^{th}$ row\footnote{$\bm{A}_i$ possessing a 2-norm is intuitive, but formally we are really considering $\bm{A}_i^T$} of $A$. When considering the variance of the sum of the entries of $\bm{D}^s$, we must also find the covariance of the terms between two elements. For $D_i^1$ and $D_j^1$, where $i \neq j$, and with $s = 1$, we have
\begin{equation*}
    \begin{split}
        \text{Cov}(D_i^1,D_j^1) &= \mathbb{E}\Big[(D_i^1 - A_{ii})(D_j^1 - A_{jj})\Big]\\
        &= \mathbb{E}\Big[\Big(\sum_{k\neq i}^n A_{ik}v^i v^k\Big)\Big(\sum_{l\neq j}^n A_{jl}v^j v^l\Big)\Big]\\
        &= \sum_{k \neq i}^n\sum_{l \neq j}^n A_{ik}A_{jl}\mathbb{E}\Big[v^i v^j v^k v^l\Big]
    \end{split}
\end{equation*}
and since $k \neq i$ and $l \neq j$, along with the fact we are examining $i \neq j$, only one term is returned: when $k = j$ and $l = i$. All other terms vanish as a result of independence and zero expectation. This leaves
\begin{equation*}
\centering
    \begin{split}
        \text{Cov}(D_i^1,D_j^1) &= A_{ij}A_{ji}\mathbb{E}\Big[(v^i)^2(v^j)^2\Big]\\
        &= A_{ij}A_{ji}\mathbb{E}\Big[(v^i)^2\Big]\mathbb{E}\Big[(v^j)^2\Big]\\
        &= A_{ij}A_{ji}.
    \end{split}
\end{equation*}
Of course, if $A$ is symmetric, this is $A_{ij}^2$. Thus, for such symmetric $A$, we have
\begin{equation}
    \begin{split}
        \Var\Big(\sum_{i = 1}^n D_i\Big) &= \sum_{i=1}^n\Var(D_i) + \sum_i^n\sum_{j\neq i}^n\text{Cov}(D_i, D_j)\\
        &= \sum_{i=1}^n(\|\bm{A}_i\|_2^2 - A_{ii}) + \sum_i^n\sum_{j \neq i}^n A_{ij}^2\\
        &= 2\Big(\|A\|_F^2 - \sum_{i = 1}^n A_{ii}^2\Big).
    \end{split}
\end{equation}
So we are returned equation \eqref{Hutchinson's original variance}, the second result of Lemma \ref{Hutchinson's Original Lemma} and the relationship of trace and diagonal estimation is clear.

 The dependence of the elements $D_i^s$ upon one another threatens to confound any further analysis. Thus, to circumvent this, we first examine individual elements of the diagonal estimator, $\bm{D}^s$. When treating all elements simultaneously, we employ general bounds and inequalities that, irrespective of dependence, may be applied. In some instances, this can result in looser bounds than might be necessary in practice. Therefore our results ought to be treated as a guide, rather than exact truths. We are left with two main questions to answer: 
\begin{enumerate}
        \item How many query vectors are needed to bound the error of the stochastic diagonal estimator in \eqref{First diagonal estimator}, and how does this change with distribution of vector entries?
    \item How does matrix structure affect estimation, and how might we improve worst case scenarios?
\end{enumerate}

\section{Definitions and Estimators} \label{sec: Definitions and Estimators}
\subsection{Definitions} \label{Definitions Section}
We state precisely what is meant by an $(\varepsilon,\delta)$-approximator in the context of diagonal estimation: for reasons which will become apparent, we introduce two such definitions. We also define what is meant by a diagonal estimator when using vectors of a given distribution.

\subsubsection{$(\varepsilon, \delta)$-approximators}

\begin{definition} \label{loose-eps-delta}
    The $i^{th}$ component, $D_i^s$, of a diagonal estimator is said to be a \textbf{row dependent $(\varepsilon, \delta)$-approximator}, if 
    \begin{equation*}
       \Pr\Bigg(|D_i^s - A_{ii}|^2 \leq \varepsilon^2\Big(\|\bm{A}_{i}\|_2^2 - A_{ii}^2\Big)\Bigg) \geq 1 - \delta
    \end{equation*}
    where $\bm{A}_i$ is the $i^{th}$ row of the matrix $A$, reshaped as a vector.
\end{definition}
\noindent The error bound in this definition may be tightened to be dependent only upon $|A_{ii}|$, giving Definition \ref{tight-eps-delta} for examination of the relative error:
\begin{definition} \label{tight-eps-delta}
    The $i^{th}$ component, $D_i^s$, of a diagonal estimator is said to be a \textbf{relative $(\varepsilon, \delta)$-approximator}, if 
    \begin{equation*}
       \Pr\Big(|D_i^s - A_{ii}| \leq \varepsilon |A_{ii}|\Big) \geq 1 - \delta
    \end{equation*}
\end{definition}
\noindent Definition \ref{tight-eps-delta} gives us a way to analyse the estimators by bounding the number of queries required to guarantee that the probability the relative error exceeds $\varepsilon$ is at most $\delta$. Definition \ref{loose-eps-delta} loosens this condition slightly and is useful for reasons which become shortly apparent.

\subsubsection{Two estimators} \label{Two estimators}
 Recall that all estimators follow the same simple pattern, as in equation \eqref{First diagonal estimator}: some random vector $\bm{v}_k$ is drawn from a fixed distribution, and a normalised Hadamard product is used to estimate the diagonals. The process is repeated $s$ times, using i.i.d queries.
 
 The first estimator uses vectors whose entries are standard Gaussian (normal) variables.

\begin{definition} \label{Gaussian definition}
    A \textbf{Gaussian diagonal estimator} for any square matrix $A \in \mathbb{R}^{n\times n}$ is given by
    \begin{equation} \label{Gaussian definition equation}
        \bm{G}^s = \Big[ \sum_{k=1}^s \bm{v}_k \odot A \bm{v}_k \Big] \oslash \Big[\sum_{k=1}^s \bm{v}_k \odot \bm{v}_k \Big]
    \end{equation}
    where the $\bm{v}_k$ vectors are $s$ independent random vectors whose entries are i.i.d standard normal variables, $v_k^i \sim N(0,1) \iff \bm{v}_k \sim N(0, I_n)$.
\end{definition}
\noindent Note that the Gaussian estimator does not constrain the 2-norm of the $\bm{v}_k$ vectors; so any single estimate may be arbitrarily large or small.

 In contrast, the Rademacher diagonal estimator constrains the $\bm{v}_k$ to have a fixed 2-norm, and results in a fixed denominator in the estimator equation.

\begin{definition} \label{Rademacher definition}
    A \textbf{Rademacher diagonal estimator} for any square matrix $A \in \mathbb{R}^{n\times n}$ is given by
    \begin{equation} \label{Rademacher definition equation}
        \bm{R}^s = \Big[ \sum_{k=1}^s \bm{v}_k \odot A \bm{v}_k \Big] \oslash \Big[\sum_{k=1}^s \bm{v}_k \odot \bm{v}_k \Big] = \frac{1}{s}\sum_{k=1}^s \bm{v}_k \odot A \bm{v}_k
    \end{equation}
    where the $\bm{v}_k$ vectors are s independent random vectors whose entries are i.i.d uniformly distributed Rademacher variables: $v_k^i \sim \{-1,1\}$ each with probability $1/2$.
\end{definition}

It is worth noting that, unlike the trace estimators, the diagonal estimators above return different values when applied to $A^T$ instead of $A$. It is possible that some (but never all, as our results indicate) diagonal elements are estimated more accurately by working with $A^T$.

\section{Comparing the Quality of Estimators} \label{Comparing the Quality of estimators.}

Bekas' et.\ al \cite{bekas2007estimator} sought to replace stochastic estimation by taking the vectors $\bm{v}_k$ to be deterministic, as empirically they found that stochastic estimation could be slow. Sections \ref{Gaussian section} and \ref{Rademacher section} show why. Table \ref{tab:initial summary table} summarises our results.
\begin{table}[h]
    \caption{\textit{Summary of results for the convergence of estimators according to the distribution of the query vectors. Note that these results do not make any assumptions on the matrix $A$.}}
    \vspace{0.2cm}
    \begin{tabular}{|p{40mm}|c|c|}
        \hline
        \textbf{Query vector distribution} & \textbf{Rademacher} &  \textbf{Gaussian}\\\hline
        \textbf{Sufficient query bound \newline for a row-dependent $(\varepsilon,\delta)$-approximator} &$2\ln(2/\delta)/\varepsilon^2$& $4\log_2(\sqrt{2}/\delta)/\varepsilon^2$\\\hline
%        \textbf{Row-dependent $(\varepsilon,\delta)$-approximator} &$2\ln(2/\delta)/\varepsilon^2$& $4\log_2(\sqrt{2}/\delta)/\varepsilon^2$\\\hline
        \textbf{Sufficient query bound for a relative $(\varepsilon,\delta)$-approximator} 
%        \textbf{Relative $(\varepsilon,\delta)$-approximator}
&${\color{white}\Bigg|}2\mathlarger{\Big(\frac{\|\bm{A}_i\|_2^2 - A_{ii}^2}{A_{ii}^2}\Big)}\ln(2/\delta)/\varepsilon^2$&
 ${\color{white}\Bigg|}4\mathlarger{\Big(\frac{\|\bm{A}_i\|_2^2 - A_{ii}^2}{A_{ii}^2}\Big)}\log_2(\sqrt{2}/\delta)/\varepsilon^2$\\\hline
        \textbf{Permitted $\varepsilon$ values} & arbitrary & $(0, 1]$\\\hline
    \end{tabular}
    \label{tab:initial summary table}
\end{table}

\section{Gaussian diagonal estimator} \label{Gaussian section}
We now prove an $(\varepsilon, \delta)$ bound for the Gaussian diagonal estimator. We find a sufficient minimum number of queries $s$, for a row-dependent $(\varepsilon,\delta)$ bound to hold.
\begin{theorem} \label{Gaussian Theorem}

    Let $G_i^s$ be the Gaussian diagonal estimator \eqref{Gaussian definition equation} of $A$, then it is sufficient to take $s$ as
    \begin{equation}\label{gaussnum_each}
        s > 4\log_2(\sqrt{2}/\delta)/\varepsilon^2
    \end{equation}
    with $\varepsilon\in (0,1]$, such that the estimator satisfies
    \begin{equation*}
        \Pr\Bigg(|G_i^s - A_{ii}|^2 \leq \varepsilon^2\Big(\|\bm{A}_i\|_2^2 - A_{ii}^2\Big)\Bigg) \geq 1 - \delta,
    \end{equation*}
    for $i = 1,\hdots, n$, where $\bm{A}_i$ is the $i^{th}$ row of the matrix, reshaped as a vector.
\end{theorem}
\noindent The proof of this theorem is best read as three sequential steps: rotational invariance, conditioned integration and function bounding. We caution the reader that the vector $\bm{u}$ is introduced that has components denoted as $v$. This notation is an unfortunate, yet unavoidable, consequence of the proof.

\begin{proof}

\subsubsection*{ \\Step 1: Rotational invariance}
    The first trick is to rewrite the coefficient of each $A_{ij},\;j \neq i$, in
    \begin{equation} \label{DinGaussianProof}
        G_i^s = A_{ii} + \sum_{j\neq i} A_{ij} \frac{\sum_{k=1}^s v_k^i v_k^j}{\sum_{k=1}^s (v_k^i)^2}.
    \end{equation}
    Let us denote\footnote{We use the notation $\bm{u}_i\in\mathbb{R}^s$ so as not to confuse ourselves with the query vectors $\bm{v}_k\in\mathbb{R}^n$. Definitively, the vector $\bm{u}_i$ is composed of the $i^{th}$ components of each query vector.} $\bm{u}_i = [v_1^i, v_2^i, ..., v_s^i]^T \in \mathbb{R}^s$ such that $\bm{u}_i \sim N(0,I_s)$. Thus
    \begin{equation*} \label{DinGaussianProof2}
        G_i^s = A_{ii} + \sum_{j\neq i} A_{ij} \frac{\bm{u}_i^T \bm{u}_j}{\|\bm{u}_i\|_2^2}.
    \end{equation*}
    Namely the coefficient of each $A_{ij},\;j\neq i$, becomes 
    \begin{equation} \label{GaussRotationInvar1}
        \frac{\bm{u}_i^T \bm{u}_j}{\|\bm{u}_i\|_2^2} =     \frac{\bm{u}_i^T Q^TQ\bm{u}_j}{\|Q\bm{u}_i\|_2^2}
    \end{equation}
    for any orthogonal rotation $Q$, $Q^TQ = I_s $. The key insight is to exploit the rotational invariance of $\bm{u}_j\sim N(0, I_s)$. Hence choose $Q$ such that 
    \begin{equation*}
        \frac{Q\bm{u}_i}{\|Q\bm{u}_i\|_2} =[1, 0, \hdots, 0]^T = \bm{e}_1 \in \mathbb{R}^{s}.
    \end{equation*}
    So the coefficients \eqref{GaussRotationInvar1} become, for all $j \neq i$,
    \begin{equation*}
        \frac{1}{\|Q\bm{u}_i\|_2}\bm{e}_1^TQ\bm{u}_j = \frac{1}{\|\bm{u}_i\|_2}\bm{e}_1^TQ\bm{u}_j = \frac{g_j}{\|\bm{u}_i\|_2}
    \end{equation*}
    since $Q$ is orthogonal, so $Q\bm{u}_j \sim N(0,I_s)\implies \bm{e}_1^TQ\bm{u}_j = g_j \sim N(0, 1)$. Note that the unit direction of a  standard normal multivariate vector and its 2-norm are independent, implying that $Q$ and thus $Q\bm{u}_j$ are independent of $\|\bm{u}_i\|_2$. Thus we may write \eqref{DinGaussianProof} as
    \begin{equation*}
        G_i^s = A_{ii} + \sum_{j \neq i}A_{ij} \frac{g_j}{\|\bm{u}_i\|_2}.
    \end{equation*}
    \noindent Let us denote $\bm{\widetilde{A}}_i \in \mathbb{R}^{n-1}$ as a vector containing the elements of the $i^{th}$ row of $A$, excluding $A_{ii}$. Also let $\bm{\widetilde{g}} \in \mathbb{R}^{n-1}$ be a vector containing the independent Gaussian $g_j$ entries. That is $\bm{\widetilde{g}} \sim N(0,I_{n-1})$. So we have
    \begin{equation*}
        G_i^s - A_{ii} = \frac{\bm{\widetilde{A}}_i^T\bm{\widetilde{g}}}{\|\bm{u}_i\|_2}.
    \end{equation*}
    The second trick is to consider
    \begin{equation}
        \frac{G_i^s - A_{ii}}{\|\bm{\widetilde{A}}_i\|_2} = \frac{1}{\|\bm{u}_i\|_2} \cdot \frac{\bm{\widetilde{A}}_i^T\bm{\widetilde{g}}}{\|\bm{\widetilde{A}}_i\|_2}.
    \end{equation}
    Once again, we use the rotational invariance of the distribution of $\bm{\widetilde{g}}$. The distribution of $\bm{\widetilde{A}}_i^T\bm{\widetilde{g}} / \|\bm{\widetilde{A}}_i\|_2$ is equal to $\bm{e}_1^T \widetilde{Q} \bm{\widetilde{g}}$, where now, $\bm{e}_1 = [1,0,...,0]^T \in \mathbb{R}^{n-1}$ and $\widetilde{Q}$ rotates $ \widetilde{\bm{A}}_i/\|\widetilde{\bm{A}}_i\|_2$ to $\bm{e}_1$. Hence we have that the distribution of $\bm{\widetilde{A}}_i^T\bm{\widetilde{g}} / \|\bm{\widetilde{A}}_i\|_2$ is equal to some $x \sim N(0,1)$, giving 
    \begin{equation} \label{ChiRatio}
        \frac{G_i^s - A_{ii}}{\|\bm{\widetilde{A}}_i\|_2} = \frac{x}{\|\bm{u}_i\|_2} \implies \frac{(G_i^s - A_{ii})^2}{\|\bm{\widetilde{A}}_i\|_2^2} = \frac{x^2}{\|\bm{u}_i\|_2^2}
    \end{equation}
    where clearly $x$ and $\|\bm{u}_i\|_2$ are independent\footnote{This a scaled F-distribution.}.
    \subsubsection*{Step 2: Conditioned integration}
    Recall that $\bm{u}_i \sim N(0,I_s) \in \mathbb{R}^s$ and so $\|\bm{u}_i\|_2^2 \sim \chi^2_s$, is a chi-squared distribution of degree $s$, and, of course, $x^2 \sim \chi^2_1$, is a chi-squared distribution of degree 1. Thus the proof reduces to bounding $\Pr(\big|x/\|\bm{u}_i\|_2 |\leq \varepsilon) =\Pr(x^2/\|\bm{u}_i\|_2^2 \leq \varepsilon^2)$ on the right hand side of equation \eqref{ChiRatio}. This is most easily examined by investigating the tail bound,
    \begin{equation} \label{F equation}
        F:=\Pr\Bigg(\frac{x^2}{\|\bm{u}_i\|_2^2} \geq \varepsilon^2\Bigg) =\Pr\Bigg(\frac{x^2}{y} \geq \varepsilon^2\Bigg)
    \end{equation}
    where, for ease of notation, we have let $y = \|\bm{u}_i\|_2^2$. The next key insight is to approach this by conditioning the probability. Thus we have,
    \begin{equation}
        \begin{split}
            F=&\Pr(x^2\geq \varepsilon^2 y)\\
            =& \int_0^\infty\Pr(x^2 \geq \varepsilon^2 y \big| y)\cdot f(y)dy
        \end{split}
    \end{equation}
    where $f(y)$ is the probability density function of the $y \sim \chi_s^2$ variable. So we have
    \begin{equation*}
        \begin{split}
            F =& \int_0^\infty\Pr(x^2 \geq \varepsilon^2 y \big| y)\cdot f(y) dy \\
            =& \int_0^\infty\Pr(\e^{\lambda x^2} \geq \e^{\lambda\varepsilon^2 y} \big| y)\cdot f(y) dy \hspace{5.2cm}\text{(for $\lambda >0$)} \\
%\hspace{7.2cm}\text{(for $\lambda >0$)}& \\
            \leq & \int_0^\infty \frac{\mathbb{E}[\e^{\lambda x^2}]}{\e^{\lambda\varepsilon^2 y}}\cdot f(y) dy \hspace{5cm}\text{(by Markov's inequality)} \\
%\hspace{7cm}\text{(by Markov's inequality)}& \\
            =& \int_0^\infty \frac{1}{\sqrt{1-2\lambda}\:\e^{\lambda\varepsilon^2 y}}\cdot f(y) dy %\hspace{3.7cm} (\lambda\in(0,1/2),\text{ MGF of a normal variable})& 
\hspace{1.7cm} (\lambda\in(0,1/2),\text{ MGF of a normal variable}) 
        \end{split}
    \end{equation*}    
    which, introducing the pdf, gives    
    \begin{equation*}
        \begin{split}
            =& \int_0^\infty \frac{1}{\sqrt{1-2\lambda}\:\e^{\lambda\varepsilon^2 y}}\cdot \frac{1}{2^{s/2}\Gamma(s/2)}y^{s/2 - 1}\e^{-\frac{1}{2}y}dy\hspace{2.2cm}\text{(pdf of $\chi_s^2$ variable)}\\
            =& \int_0^\infty \frac{1}{\sqrt{1-2\lambda}}\cdot \frac{1}{2^{s/2}\Gamma(s/2)}y^{s/2 - 1}\e^{-(\frac{1}{2}+\lambda\varepsilon^2)y}dy\hspace{3.4cm}\text{(rearrange).}
        \end{split}
    \end{equation*}
    
 Let $z = (1 + 2\lambda\varepsilon^2)y$, such that $dz = (1 + 2\lambda\varepsilon^2)dy$. Thus substitution yields
    
    \begin{equation*}
        \begin{split}
            F \leq& \int_0^\infty \frac{1}{\sqrt{1-2\lambda}}\cdot\frac{1}{(1+2\lambda\varepsilon^2)^{s/2}}\cdot \frac{1}{2^{s/2}\Gamma(s/2)}z^{s/2 - 1}\e^{-\frac{1}{2}z}dz\\
            =& \frac{1}{\sqrt{1-2\lambda}}\cdot\frac{1}{(1+2\lambda\varepsilon^2)^{s/2}}\hspace{4cm}\text{(by the integral of the pdf)}\\
            =& \frac{1}{\sqrt{1-2\lambda}}\cdot\frac{1}{\e^{\ln(1+2\lambda\varepsilon^2)s/2}}.
        \end{split}
    \end{equation*}
    \subsubsection*{Step 3: Function bounding}
    Note that $2\lambda\varepsilon^2 \in (0,1)$ since $\lambda \in (0,1/2)$ and $\varepsilon \in (0,1]$, and consider that
    \begin{equation*} \label{ln bound}
        \ln(1+x) \geq \ln(2)x
    \end{equation*}
    for $x \in (0,1)$. Therefore we can bound $F$ by
    \begin{equation} \label{F bound}
        \begin{split}
            F \leq \frac{1}{\sqrt{1-2\lambda}}\cdot\frac{1}{\e^{\ln(2)\lambda\varepsilon^2s}} = \frac{1}{\sqrt{1-2\lambda}}\cdot\frac{1}{2^{\lambda\varepsilon^2s}}
        \end{split}
    \end{equation}
    Naively (we do not deal with optimal $\lambda$ here), setting $\lambda = 1/4$, and bounding the above by $\delta$ gives
    \begin{equation} \label{Gaussian Proof End}
        \begin{split}
            F\leq \frac{\sqrt{2}}{2^{s\varepsilon^2/4}} < \delta
        \end{split}
    \end{equation}
    thus we have 
    \begin{equation}
        s > 4\log_2\Big(\frac{\sqrt{2}}{\delta}\Big) \Big/ \varepsilon^2
    \end{equation}
    as a sufficient bound for $s$, such that, for all such $s$
    \begin{equation}
        \Pr(|G_i^s - A_{ii}| \geq \varepsilon \|\bm{\widetilde{A}}_i\|_2) = F < \delta
    \end{equation}
    which completes the proof.
\hfill$\square$
\end{proof}

%\begin{center}
%\rule{10cm}{0.4pt}
%\end{center}

\subsubsection*{Implications}
\noindent Theorem \ref{Gaussian Theorem} deals only with a single entry of the diagonal. Concretely, we have found that 
\begin{equation*} \label{Concrete post Gaussian}
    \begin{split}
        \Pr\Big(|G_i^s - A_{ii}|^2 \leq \varepsilon^2 \|\bm{\widetilde{A}}_i\|_2^2\Big) \geq 1 -\delta
    \end{split}
\end{equation*}
for $s > 4\log_2(\sqrt{2}/\delta)/\varepsilon^2$. This means $G_i^s$ is a row-dependent $(\varepsilon,\delta)$-approximator for all such $s$. For a relative $(\varepsilon,\delta)$-approximator, we must substitute
\begin{equation*}
    \varepsilon^2 = \bar{\varepsilon}^2\frac{A_{ii}^2}{\|\bm{\widetilde{A}}_i\|_2^2}
\end{equation*}
to get
\begin{equation*}
    \Pr\Big(|G_i^s - A_{ii}| \leq \bar{\varepsilon}|A_{ii}|\Big) \geq 1 - \delta
\end{equation*}
provided
\begin{equation} \label{Gauss single s}
    s > 4\Bigg(\frac{\|\bm{\widetilde{A}}_i\|_2^2}{A_{ii}^2}\Bigg)\frac{\log_2(\sqrt{2}/\delta)}{\bar{\varepsilon}^2}.
\end{equation}

\noindent For the entire diagonal, if we demand $|G_i^s - A_{ii}|^2 \leq \varepsilon^2\Big(\|\bm{A}_i\|_2^2 - A_{ii}^2\Big)\hspace{0.2cm}\forall\; i$ we may apply the ideas from the proof of Theorem \ref{Gaussian Theorem}, to obtain the number of vectors required for this condition to hold.
\begin{corollary} \label{Gaussian Union Theorem}
    For
    \begin{equation*}
        \Pr\Bigg(\sum_{i=1}^n|G_i^s - A_{ii}|^2 \leq \varepsilon^2\sum_{i=1}^n\Big(\|\bm{A}_i\|_2^2 - A_{ii}^2\Big)\Bigg) \geq 1 - \delta
    \end{equation*}
    with $\varepsilon\in (0,1]$ it is sufficient to take
    \begin{equation*}\label{cor:gauvec}
%\rrr{        s = O(\log(n/\delta)/\varepsilon^2)}{
s> 4\log_2\Big(\frac{n\sqrt{2}}{\delta}\Big) \Big/ \varepsilon^2
    \end{equation*}
    query vectors if vector entries are i.i.d Gaussian, where $\bm{A}_i$ is the $i^{th}$ row of the matrix A, and $G_i^s$ is the $i^{th}$ component of the Gaussian diagonal estimator.
\end{corollary}
\begin{proof}
By equation \eqref{Gaussian Proof End} of the proof of Theorem \ref{Gaussian Theorem}, we seek
\begin{equation*}
   \Pr\Bigg(\frac{(G_i^s - A_{ii})^2}{\|\bm{\widetilde{A}}_i\|_2^2} \geq \varepsilon^2\Bigg) = \Pr\Bigg(\frac{\chi^2_1}{\chi^2_s} \geq \varepsilon^2\Bigg)\leq \frac{\sqrt{2}}{2^{s\varepsilon^2/4}}\;\forall\; i.
\end{equation*}
Applying a union bound, and bounding this by $\delta$ we readily find
\begin{equation} \label{s-whole diagonal gauss}
    n\cdot\frac{\sqrt{2}}{2^{s\varepsilon^2/4}} < \delta \implies s > 4\log_2\Big(\frac{n\sqrt{2}}{\delta}\Big) \Big/ \varepsilon^2
\end{equation}
is sufficient, which completes the proof.
\hfill$\square$
\end{proof}
\noindent Taking the summation over $i$, Corollary \ref{Gaussian Union Theorem} means we have, with probability at least $1 - \delta$
\begin{equation}
    \|\bm{G}^s - \bm{A}_d\|_2^2 \leq \varepsilon^2\big(\|A\|_F^2 - \|\bm{A}_d\|_2^2\big)
\end{equation}
if we have $s$ as in equation \eqref{s-whole diagonal gauss}.
To make this bound a relative estimate, in comparison to $\bm{A}_d$, we write 
\begin{equation*}
    \varepsilon^2 = \Bigg(\frac{\|A\|_F^2 - \|\bm{A}_d\|_2^2}{\|\bm{A}_d\|^2}\Bigg)\bar{\varepsilon}^2
\end{equation*}
to get 
\begin{equation}
    \Pr\Big(\|\bm{G}^s -\bm{A}_d\|_2 \leq \bar{\varepsilon}\|\bm{A}_d\|_2\Big) \geq 1 - \delta
\end{equation}
provided that the number of query vectors $s$ is
\begin{equation} \label{s-queries full diagonal}
    s > 4\Bigg(\frac{\|A\|_F^2 - \|\bm{A}_d\|_2^2}{\|\bm{A}_d\|^2}\Bigg)\frac{\log_2\Big(n\sqrt{2}/\delta\Big)}{\bar{\varepsilon}^2}.
\end{equation}
In summary, we have found that to get a relative 2-norm estimate of the entire diagonal with error at most $\bar{\varepsilon}$, and probability at least $1- \delta$, it is sufficient to use the $s$ queries as in equation \eqref{s-queries full diagonal}, if Gaussian query vectors are used.

\section{Rademacher diagonal estimator} \label{Rademacher section}
This section deals with the proof of an $(\varepsilon, \delta)$ bound for the Rademacher diagonal estimator. We choose to break this down into two main steps, the first of which is given in Lemma \ref{Chernoff Lemma}. The lemma concerns the expectation of an exponential, 
which will subsequently be used in a Chernoff-style argument. One may regard it simply as a slight tweak on Khintchine's lemma, but introducing an additional summation term and so we outline it fully for completeness.
Before continuing, we note that examining the Rademacher estimator element-wise we have
\begin{equation} \label{For chapter 5}
    \begin{split}
    R_i^s = A_{ii} + \sum_{j \neq i}^n A_{ij} \frac{\sum_{k=1}^s v_k^iv_k^j}{s}
    = A_{ii} + \frac{1}{s}\sum_{j\neq i}^n A_{ij} \sum_{k=1}^s z_k^{ij}
    \end{split}
\end{equation}
where we have rewritten $v_k^i v_k^j$ as $z_k^{ij}$ and $z_k^{ij}$ is also a Rademacher variable\footnote{It is easy to verify that a Rademacher variable times an independent ($j \neq i$) Rademacher variable is also Rademacher variable.}. We now introduce the following lemma.

\begin{lemma} \label{Chernoff Lemma}
    For 
    \begin{equation*}
        X^s_i := R_i^s - A_{ii}
    \end{equation*}
    we have that 
    \begin{equation*}
        \mathbb{E}[\exp(t X^s_i)] \leq \exp\Bigg(\frac{t^2 \|\bm{\widetilde{A}}_i\|_2^2}{2s}\Bigg)
    \end{equation*}
    where $\bm{\widetilde{A}}_i \in \mathbb{R}^{n-1}$ is a vector containing the elements of the $i^{th}$ row of A, excluding the $i^{th}$ element.
\end{lemma}
\begin{proof}
    Note that $X^s_i$ is simply
    \begin{equation*}
        X_i^s = \frac{1}{s}\sum_{j\neq i}^n A_{ij} \sum_{k=1}^s z_k^{ij}
    \end{equation*}
    Consider now
    \begin{equation*}
        \begin{split}
            \mathbb{E}[\exp(t X_i^s)] =& \mathbb{E}\Big[\exp \Big(\frac{t}{s}\sum_{j\neq i}^n A_{ij} \sum_{k=1}^s z_k^{ij}\Big) \Big]\\
            =& \prod_{j \neq i}^n \prod_{k=1}^s \mathbb{E}\Big[\exp \Big(\frac{t}{s}A_{ij}z_k^{ij}\Big)\Big] \\
            =& \prod_{j \neq i}^n \prod_{k=1}^s \frac{1}{2}\Big(\exp \Big(\frac{t}{s}A_{ij}\Big) +\Big(-\frac{t}{s}A_{ij}\Big)\Big)\\
            =& \prod_{j \neq i}^n \prod_{k=1}^s \Big(1 + \frac{(t A_{ij}/s)^2}{2!} + \frac{(t A_{ij}/s)^4}{4!} + ... \Big)\\
            \leq & \prod_{j \neq i}^n \prod_{k=1}^s \Big(1 + \frac{(t A_{ij}/s)^2}{1!\cdot 2} + \frac{(t A_{ij}/s)^4}{2!\cdot 2^2} + ... \Big)\\
            =& \prod_{j \neq i}^n \prod_{k=1}^s \exp \Big( \frac{t^2 A_{ij}^2}{2s^2} \Big)\\
            =& \prod_{j \neq i}^n \exp \Big(\frac{t^2 A_{ij}^2}{2 s}\Big)\\
            =& \exp \Big(\frac{t^2\|\bm{\widetilde{A}}_i\|_2^2}{2s}\Big)
        \end{split}
    \end{equation*}
    which completes the proof.
\hfill$\square$
\end{proof}

\noindent With this lemma in place, we can now prove an $(\varepsilon, \delta)$ bound for the Rademacher diagonal estimator. Specifically, we find the minimum sufficient number of queries $s$, for a row-dependent $(\varepsilon,\delta)$ bound to hold.

\begin{theorem} \label{Rademacher Theorem}
    Let $R_i^s$ be the Rademacher diagonal estimator \eqref{Rademacher definition equation} of $A$, then it is sufficient to take $s$ as
    \begin{equation}\label{rad_numeach}
        s > 2\ln(2/\delta)/\varepsilon^2
    \end{equation}
    with arbitrary $\varepsilon$, such that the estimator satisfies
    \begin{equation*}
        \Pr\Bigg(|R_i^s - A_{ii}|^2 \leq \varepsilon^2\Big(\|\bm{A}_i\|_2^2 - A_{ii}^2\Big)\Bigg) \geq 1 - \delta,
    \end{equation*}
    where $\bm{A}_i$ is the $i^{th}$ row of the matrix reshaped as a vector.
\end{theorem}
\begin{proof}
    The problem may be reduced to finding the tail probability, $\Pr(|X_i^s| \geq \varepsilon \|\widetilde{\bm{A}}_i\|_2)$. 
    Investigating the positive tail probability we have that
    \begin{equation} \label{Positive tail argument}
        \begin{split}
           \Pr(X^s_i \geq \varepsilon \|\widetilde{\bm{A}}_i\|_2) =&\Pr(\e^{t X^s_i} \geq \e^{t\varepsilon\|\widetilde{\bm{A}}_i\|_2})\hspace{3.1cm}\text{(all $t >0$)}\\
            \leq & \frac{\mathbb{E}[\e^{tX^s_i}]}{\e^{t\varepsilon\|\widetilde{\bm{A}}_i\|_2}}\hspace{4.9cm}\text{(By Markov's inequality)}\\
            \leq& \exp \Big(-t\varepsilon \|\widetilde{\bm{A}}_i\|_2 + t^2 \|\widetilde{\bm{A}}_i\|_2^2/2s \Big).\hspace{1.2cm}\text{(By Lemma \ref{Chernoff Lemma})}
        \end{split}
    \end{equation}
    This is minimum for $t = s\varepsilon/\|\widetilde{\bm{A}}_i\|_2$, so we obtain
    \footnote{We can also obtain~\eqref{plusXgreater} from a direct application of Hoeffding's bound for the sum of unbounded random variables~\cite[Prop.~2.5]{wainwright2019high}.}
    \begin{equation} \label{plusXgreater}
       \Pr(X^s_i \geq \varepsilon \|\widetilde{\bm{A}}_i\|_2) \leq \exp\Big(-\frac{1}{2} s \varepsilon^2\Big).
    \end{equation}
    Since we are considering $|X^s_i| \geq \varepsilon \|\widetilde{\bm{A}}_i\|_2$, we need also to bound the negative tail
    \begin{equation*}
        \begin{split}
            \Pr(-X^s_i \geq \varepsilon \|\widetilde{\bm{A}}_i\|_2) &=\Pr(\e^{-t X^s_i} \geq \e^{ t \varepsilon \|\widetilde{\bm{A}}_i\|_2})\\
            &\leq \frac{\mathbb{E}[\e^{-t X^s_i}]}{\e^{t\varepsilon\|\widetilde{\bm{A}}_i\|_2}}.
       \end{split}
    \end{equation*}
    It is also easily shown, by Lemma \ref{Chernoff Lemma}, that
    \begin{equation*}
        \mathbb{E}[\e^{-t X^s_i}] \leq \exp((-t)^2\|\widetilde{\bm{A}}_i\|^2_2/2 s)
    \end{equation*}
    which means that, just as for the positive tail
    \begin{equation} \label{minusXgreater}
       \Pr( -X^s_i \geq \varepsilon \|\widetilde{\bm{A}}_2\|_2) \leq \exp\Big(-\frac{1}{2} s \varepsilon^2\Big)
    \end{equation}
    Or indeed, we may easily deduce the above by considering the symmetry of the Rademacher variables $z_k^{ij}$ in Lemma \ref{Chernoff Lemma}.\\
    
    \noindent Given that the events in equations \eqref{plusXgreater} and \eqref{minusXgreater} are disjoint, we have
    \begin{equation} \label{Rademacher Proof Penultimate}
       \Pr(|X^s_i| \geq \varepsilon\|\widetilde{\bm{A}}_i\|_2) \leq 2\exp \Big( - \frac{1}{2} s \varepsilon^2 \Big).
    \end{equation}
    Bounding this probability by $\delta$ yields
    \begin{equation} 
    2\exp \Big( - \frac{1}{2} s \varepsilon^2 \Big) < \delta \implies
        s > 2 \ln\Big(\frac{2}{\delta}\Big)\Big/\varepsilon^2
    \end{equation}
    as the sufficient lower bound for $s$, such that, for all such $s$
  \begin{equation}
      \Pr(|X_i^s|\geq\varepsilon\|\widetilde{\bm{A}}_i\|_2) < \delta
  \end{equation} which completes the proof.
\hfill$\square$
\end{proof}
%\begin{center}
%    \rule{10cm}{0.4pt}
%\end{center}

\subsubsection*{Implications}
\noindent The above result is for a single diagonal entry. We have found that 
\begin{equation*} \label{Concrete Post Rademacher}
    \Pr\Big(|R_i^s - A_{ii}|^2 \geq \varepsilon^2\|\widetilde{\bm{A}}_i\|_2^2\Big) \leq 2\exp(-s\varepsilon^2/2) < \delta
\end{equation*}
so that $R_i^s$ is a row-dependent $(\varepsilon,\delta)$-approximator for $s > 2\ln(2/\delta)/\varepsilon^2$. For a relative $(\varepsilon,\delta)$-approximator we must substitute\footnote{Replacing $\|\bm{\widetilde{A}}_i\|_2$ with $A_{ii}$ in equation \eqref{Positive tail argument}, yields the same result as in equation \eqref{single element rademacher queries}.}
\begin{equation*} 
    \varepsilon^2 = \bar{\varepsilon}^2\frac{A_{ii}^2}{\|\bm{A}_d\|_2^2}
\end{equation*}
to get 
\begin{equation} \label{single element rademacher probability}
    \Pr\Big(|R_i^s - A_{ii}| \geq \bar{\varepsilon} |A_{ii}|\Big) < \delta
\end{equation}
provided
\begin{equation} \label{single element rademacher queries}
    s > 2\Bigg(\frac{\|\widetilde{\bm{A}}_i\|_2^2}{A_{ii}^2}\Bigg)\frac{\ln(2/\delta)}{\bar{\varepsilon}^2}.
\end{equation}
For the whole diagonal, and in much the same way as when considering the Gaussian case, if we demand $|R_i^s - A_{ii}|^2 \leq \varepsilon^2\Big(\|\bm{A}_i\|_2^2 - A_{ii}^2\Big)\hspace{0.2cm}\forall \; i$, we need simply apply the ideas from the proof of Theorem \ref{Rademacher Theorem}, to obtain the number of vectors required for this condition to hold. 

\begin{corollary} \label{Rad Union Theorem}
For
\begin{equation*}
    \Pr\Bigg(\sum_{i=1}^n|R_i^s - A_{ii}|^2 \leq \varepsilon^2\sum_{i=1}^n\Big(\|\bm{A}_i\|_2^2 - A_{ii}^2\Big)\Bigg) \geq 1 - \delta
\end{equation*}
with $\varepsilon$ arbitrary it is sufficient to take
\begin{equation}\label{eq:radnum}
%\rrr{    s = O(\log(n/\delta)/\varepsilon^2)}{
s>2 \ln\Big(\frac{2n}{\delta}\Big)\Big/\varepsilon^2
\end{equation}
query vectors if vector entries are i.i.d Rademacher, where $\bm{A}_i$ is the $i^{th}$ row of the matrix A, and $R_i^s$ is the $i^{th}$ component of the Rademacher diagonal estimator.
\end{corollary}
\begin{proof}
By equation \eqref{Rademacher Proof Penultimate} in the proof of Theorem \ref{Rademacher Theorem}, we seek
\begin{equation*}
   \Pr(|X^s_i| \geq \varepsilon\|\widetilde{\bm{A}}_i\|_2) \leq 2\exp \Big( - \frac{1}{2} s \varepsilon^2 \Big)\; \forall \; i.
\end{equation*}
 Applying a union bound for all entries, and bounding this, we readily find 
\begin{equation} \label{s queries for Rademacher diagonal}
    n\cdot2\exp \Big( - \frac{1}{2} s \varepsilon^2 \Big) < \delta \implies s > 2 \ln\Big(\frac{2n}{\delta}\Big)\Big/\varepsilon^2
\end{equation}
which completes the proof.
\hfill$\square$
\end{proof}

\noindent Thus taking the summation over $i$, we have, with probability $1-\delta$
\begin{equation} \label{s-whole diagonal rad}
    \|\bm{R}^s - \bm{A}_d\|_2^2 \leq \varepsilon^2\big(\|A\|_F^2 -\|\bm{A}_d\|_2^2\big)
\end{equation}
provided we have used $s$-queries as in \eqref{s queries for Rademacher diagonal}. If we seek a relative error, in comparison to $\bm{A}_d$, let 
\begin{equation*}
    \varepsilon^2 = \Bigg(\frac{\|A\|_F^2 - \|\bm{A}_d\|_2^2}{\|\bm{A}_d\|_2}\Bigg)\bar{\varepsilon}^2
\end{equation*}
to get 
\begin{equation}
    \Pr\Big(\|\bm{R}^s - \bm{A}_d\|_2 \leq \bar{\varepsilon}\|\bm{A}_d\|_2\Big) \geq 1 - \delta
\end{equation}
provided
\begin{equation} \label{s-queries full diagonal rad}
     s > 2\Bigg(\frac{\|A\|_F^2 - \|\bm{A}_d\|_2^2}{\|\bm{A}_d\|^2}\Bigg)\frac{\ln\Big(2n/ \delta\Big)}{\bar{\varepsilon}^2}
\end{equation}
Overall, we have found that to get a relative 2-norm estimate of the entire diagonal with error at most $\varepsilon$, and probability at least $1-\delta$, it is sufficient to use the $s$ queries as in equation \eqref{s-queries full diagonal rad}, if we use Rademacher query vectors.

\subsubsection*{A brief comparison}
Comparing~\eqref{gaussnum_each} with~\eqref{rad_numeach} (or~\eqref{cor:gauvec} with~\eqref{eq:radnum}), we see that both the Rademacher estimator and the Gaussian estimator display strikingly similar criteria for the selection of $s$. Whilst the Gaussian estimator appears slightly worse than its Rademacher counterpart, this difference is, in general terms, very small. A more insidious point to note relates to the demands on $\varepsilon$ in the Gaussian case. We have demanded $\varepsilon \in (0,1]$ for equation \eqref{F bound} to hold. Whilst seemingly a minor detail, this becomes very important if we are to use a small number of query vectors $s$, say, on the order of $10$. In Section \ref{Numerical Experiments on the diagonal} we demonstrate the implications of this assumption, and offer intuition as for why, when $s$ is small, the Gaussian diagonal estimator performs worse than the Rademacher diagonal estimator.

\section{Positive-semi-definite diagonal estimation} \label{Positive-semi-definite diagonal estimation}

We now examine the case where $A\succeq 0$, as is relevant in many applications \cite{braverman2020schatten,chen2016accurately,cohen2018approximating,di2016efficient,han2017approximating,li2020well,lin2016approximating,ubaru2017applications}, and analysed extensively in \cite{avron2011randomized,roosta2015improved}. For both estimators we have
\begin{equation} \label{Diagonal column bound equation}
    |D_i^s - A_{ii}|^2 \leq \varepsilon^2\Big(\|\bm{A}_i\|_2^2 - A_{ii}^2\Big)
\end{equation}
with probability, $1-\delta$, provided we have used $s = O(\log(1/\delta)/\varepsilon^2)$ queries and for given $i$. As before, $\bm{A}_i$ is the $i^{th}$ row of the matrix (and therefore the $i^{th}$ column by symmetry). There is no escaping the dependence of the absolute error on the matrix structure, and, as we have also seen, shifting to a relative error framework means that the number of sufficient queries becomes dependent on the matrix. Can we draw any other general conclusions about the dependence of the number of queries on the structure of the matrix? Might they fit with our intuitions, and indicate when stochastic estimation may be suitable? We find, as expected, that both flatter spectra lend themselves to diagonal estimation and that the eigenvector basis of the matrix plays a key role. These results are included for the practitioner, who may have access to some of the quantities explored in the following sections (e.g. $\kappa_2(A)$), or may be able to simulate other quantities in particular scenarios (e.g. $\sigma_{\min}(V\odot V)$). 

\subsection{Spectrum influence - element-wise}
To make a start let us consider the following lemma. Recall that $\lambda_1$ is the largest eigenvalue.
\begin{lemma}\label{Column bound lemma}
    For all $A \succeq 0$ we have that $\|\bm{A}_i\|_2^2 \leq \lambda_1 A_{ii} \hspace{0.2cm} \forall\; i$.
\end{lemma}
\begin{proof} \label{Column bound proof}
   We may express the column $\bm{A}_i$ as
   \begin{equation*}
       \bm{A}_i = \sum_{j=1}^n \lambda_j V_{ij} \bm{V}_j
   \end{equation*}
   where $\bm{V}_j$ is the $j^{th}$ eigenvector of $A$ and so we also have that $ A_{ii} = \sum_{j=1}^n \lambda_j (V_{ij} )^2$. Thus
   \begin{equation*} 
   \|\bm{A}_i\|_2^2 = \Big\|\sum_{j=1}^n\lambda_j V_{ij} \bm{V}_j \Big\|_2^2 = \sum_{j=1}^n (\lambda_jV_{ij} )^2 \leq \lambda_1 \sum_{j=1}^n \lambda_j (V_{ij})^2 = \lambda_1 A_{ii}
   \end{equation*}
   where the second equality holds by orthogonality, and the inequality by positive-semi-definiteness.
\hfill$\square$
\end{proof}

\begin{remark} \label{Zero-Remark}
    For $A \succeq 0$, where $\lambda_n = 0$, then $A_{ii}$ may take the value $0$. When $A_{ii}$ does take this value, by Lemma \ref{Column bound lemma} we have that $\bm{A}_i = \bm{0}$ and thus all diagonal zeroes will be found exactly by either the Rademacher or Gaussian estimator.
\end{remark}
 Now suppose $A\succ 0$, such that $\lambda_n \neq 0$, then by equation \eqref{Diagonal column bound equation} and Lemma \ref{Column bound lemma}, for sufficient $s$ we have
\begin{equation} \label{kappa2 derivation}
    \begin{split}
        |D_i^s - A_{ii}|^2 &\leq \varepsilon^2(\lambda_1 A_{ii} - A_{ii}^2)\\
        & = \varepsilon^2\Big(\frac{\lambda_1}{A_{ii}} - 1\Big)A_{ii}^2\\
        &\leq \varepsilon^2(\kappa_2(A) - 1)A_{ii}^2.
    \end{split}
\end{equation}
Then, setting $\bar{\varepsilon}^2 = \varepsilon^2/(\kappa_2(A) - 1)$, we shall return to a relative $(\varepsilon,\delta)$-approximator provided
\begin{equation} \label{query kappa bound}
    s > O\Bigg(\frac{(\kappa_2(A) - 1)\log(1/\delta)}{\bar{\varepsilon}^2}\Bigg).
\end{equation}
Extending this to all diagonal elements, we have the following lemma.
\begin{lemma} \label{conditioning lemma}
For
\begin{equation*}
    \Pr\Big(\|\bm{D}^s - \bm{A}_d\|_2^2 \leq \bar{\varepsilon}^2 \|\bm{A}_d\|_2^2 \Big) \geq 1-\delta 
\end{equation*}
we have that
\begin{equation} \label{query kappa bound diag}
    s > O\Bigg(\frac{(\kappa_2(A) - 1)\log(n/\delta)}{\bar{\varepsilon}^2}\Bigg)
\end{equation}
is a sufficient number of queries for both the Rademacher or Gaussian estimator.
\end{lemma}
\noindent Of course, this analysis produces a weaker sufficiency bound\footnote{Due to the inequalities applied, we expect to arrive at sufficiency before $s$ as in \eqref{query kappa bound diag}.} on $s$ than either of those in Table \ref{tab:initial summary table}, but it is more likely to be practically computable (e.g. via few steps of the Lanczos algorithm) or perhaps even known for implicit $A$.

 Lemma \ref{conditioning lemma} suggests roughly that the conditioning of the matrix indicates its ability to be estimated. For a given $s$, if $\kappa_2(A)$ is larger, error estimates are likely to be worse, but as $\kappa_2(A)$ approaches 1 they become better and better. This is intuitive, since if $\kappa_2(A) = 1$, then the matrix $A$ shall be a scalar multiple of the identity: $A = V(c I)V^T = c I$ for some scalar $c$. Hence the lower bound on $s$ becomes $0$ and exact estimates are found after a single query. More generally, if $\kappa_2(A) \ll n$, then the bounds \eqref{query kappa bound} and \eqref{query kappa bound diag} are good, and stochastic estimation is suitable to estimate the diagonal. Conversely, if $\kappa_2(A) = O(n)$ or greater, then clearly this bound is no longer helpful since we seek $s \ll n$ queries.

 However, if the bound of Lemma \ref{Column bound lemma} is loose, and line three of equation \eqref{kappa2 derivation} is loose, then the bounds \eqref{query kappa bound} and \eqref{query kappa bound diag} are too pessimistic. To partly rectify this, and extend these ideas for general $A\succeq 0$, we stop at the second line of \eqref{kappa2 derivation} to write, with $c_i(A) := \lambda_1/A_{ii}$
\begin{equation*}
    |D_i^s - A_{ii}|^2 \leq \varepsilon^2(c_i(A) - 1)A_{ii}^2
\end{equation*}
where $c_i(A)$ is best considered as the ``conditioning" of a diagonal element. Note again that in this analysis we are considering $A_{ii} \neq 0$, since we already know that if $A_{ii} = 0$ it will be found exactly. Let us define
\begin{equation}  \label{kappa-i intro}
 \kappa_d(A) := \frac{\lambda_1}{\min_{i,A_{ii}\neq 0}A_{ii}}
\end{equation}
as the conditioning of the diagonal. So we have, 
\begin{equation*}
    |D_s^i - A_{ii}|^2 \leq \varepsilon^2(\kappa_d(A) - 1)A_{ii}^2
\end{equation*}
and the same ideas as for $\kappa_2(A)$ may be applied. In particular, note that we shall obtain a relative $(\varepsilon,\delta)$-approximator if, with $\bar{\varepsilon}^2 = \varepsilon^2/(\kappa_d(A) - 1)$
\begin{equation}
    s > O\Bigg(\frac{(\kappa_d(A) - 1)\log(1/\delta)}{\bar{\varepsilon}^2}\Bigg)
\end{equation}
and, extending this to the diagonal elements, we readily obtain the following lemma.
\begin{lemma} \label{diagonal conditioning lemma}
For 
\begin{equation*}
    \Pr\Big(\|\bm{D}^s - \bm{A}_d\|_2^2 \leq \bar{\varepsilon}^2 \|\bm{A}_d\|_2^2\Big) \geq 1 - \delta
\end{equation*}
we have that
\begin{equation}
    s > O\Bigg(\frac{(\kappa_d(A) - 1)\log(n/\delta)}{\bar{\varepsilon}^2}\Bigg)
\end{equation}
is a sufficient number of queries for both Rademacher and Gaussian estimators.
\end{lemma}
\noindent So if $\kappa_d(A) \ll n$ then stochastic estimation is suitable to estimate the diagonal of such a matrix. Since $\kappa_2(A) \geq \kappa_d(A)$, this demand is easier to satisfy than demanding that $\kappa_2(A) \ll n$.

\subsection{Spectrum influence - the entire diagonal}
\noindent So far we have built our understanding by considering single element estimates, and examining their implications for the whole diagonal. We can gain more information by examining the inequality 
\begin{equation}
    \|\bm{D}^s - \bm{A}_d\|_2^2 \leq \varepsilon^2 (\|A\|_F^2 - \|\bm{A}_d\|_2^2)
\end{equation}
since this accounts for the entire matrix structure. This holds, with probability at least $1-\delta$ if $s > O\Big(\log(n/\delta)/\varepsilon^2\Big)$. Recall that, for $A\succeq 0$, $\|A\|_F^2 = \|\bm{\lambda}\|_2^2$, so we have
\begin{equation} \label{Ad and lambda}
    \|\bm{D}^s -\bm{A}_d\|_2^2 \leq \varepsilon^2\Big(\|\bm{\lambda}\|_2^2 - \|\bm{A}_d\|_2^2\Big).
\end{equation}
Now, we also have for $A\succeq 0$, that 
\begin{equation*}
        \|\bm{A}_d\|_2^2 \geq \frac{1}{n}\|\bm{A}_d\|_1^2 = \frac{1}{n}\|\bm{\lambda}\|_1^2
\end{equation*}
which means that 
\begin{equation} \label{lambda2 vs lambda1}
    \|\bm{D}^s - \bm{A}_d\|_2^2 \leq \varepsilon^2\Big(\|\bm{\lambda}\|_2^2 - \frac{1}{n}\|\bm{\lambda}\|_1^1\Big)
\end{equation}
and we can once again see the influence of the eigenvalues on the error of our estimates.

 This equation, however, is dependent on the entire spectrum, in comparison to only $\kappa_2(A)$ or $\kappa_d(A)$ as was previously found. As expected, the conditioning of a matrix does not tell the whole story. If $\bm{\lambda} = [1, 1, \hdots, 1, \epsilon]$ the condition number of the matrix in question may be large, even $O(n)$, but equation \eqref{lambda2 vs lambda1} suggests that the scaling of the absolute error will be acceptable. Clearly, a flat spectrum, such as $\bm{\lambda} = [1, 1, ..., 1]\in\mathbb{R}^n$, corresponding to the identity, yields a right hand side of zero. Meanwhile, the value of $\|\bm{\lambda}\|_2^2 - \|\bm{\lambda}\|_1^2/n$ is greatest when mass is concentrated on a single eigenvalue.

 For completeness, note that the reverse is true of the mass of the diagonal, since 
\begin{equation*}
    \begin{split}
        \|\bm{\lambda}\|_2^2 &\leq \|\bm{\lambda}\|_1^2 = \|\bm{A}_d\|_1^2    
    \end{split}
\end{equation*}
and so 
\begin{equation} \label{hint}
    \|\bm{D}^s - \bm{A}_d\|_2^2 \leq \varepsilon^2\Big(\|\bm{A}_d\|_1^2 - \|\bm{A}_d\|_2^2\Big)
\end{equation}
which is worse as diagonal elements become flatter but better as most diagonal mass is concentrated on a handful of entries.

\subsection{Eigenvector influence}
Equation \eqref{hint} hints at the notion that it is not only the spectrum of $A$ that influences the error of diagonal estimation. Indeed, if we have a very steep spectrum, but the eigenvectors of $A$ are columns of the identity, that is $V = I$, then we shall still have a diagonal matrix and so expect one-shot estimation. Consider the following.
\begin{lemma} \label{V hadamard lemma}
    Let $A\in\mathbb{R}^{n\times n}$ be a diagonalisable matrix, $A = V\Lambda V^T$, with $\bm{\lambda} = \textnormal{diag}(\Lambda)$, and $\bm{A}_d = \textnormal{diag}(A)$, then
    \begin{equation*}
        \bm{A}_d = (V \odot V)\bm{\lambda}
    \end{equation*}
    where $V\odot V$, is the Hadamard (element-wise) product of $V$ with itself.
\end{lemma}
\begin{proof}
    The above is easily found by recalling
    \begin{equation*}
        A_{ii} = \sum_{j=1}^n \lambda_j (V_{ij})^2 \quad \forall \; i
    \end{equation*}
    and simply extending to the matrix equation.
\hfill$\square$
\end{proof}
\noindent Hence, Lemma \ref{V hadamard lemma} and equation \eqref{Ad and lambda} imply, for $s > O\Big(\log(n/\delta)/\varepsilon^2\Big)$, \\
\begin{equation*}
    \begin{split}
        \|\bm{D}^s -\bm{A}_d\|_2^2 &\leq \varepsilon^2\Big(\|\bm{\lambda}\|_2^2 - \|(V\odot V)\bm{\lambda}\|_2^2\Big)\\
        &\leq\varepsilon^2\Big(1 - \sigma^2_{\min}(V\odot V)\Big)\|\bm{\lambda}\|_2^2.
    \end{split}
\end{equation*}
If we write
\begin{equation*}
    \varepsilon^2 = \bar{\varepsilon}^2\frac{\|\bm{A}_d\|_2^2}{\Big(1 - \sigma^2_{\min}(V\odot V)\Big)\|\bm{\lambda}\|_2^2}
\end{equation*}
so we may readily obtain the following lemma.
\begin{lemma}
For
\begin{equation*}
    \|\bm{D}^s - \bm{A}_d\|_2^2 \leq \bar{\varepsilon}^2\|\bm{A}_d\|_2^2
\end{equation*}
we have that
\begin{equation} \label{sigma min bound}
    s > O\Bigg(\frac{\Big(1 - \sigma^2_{\min}(V\odot V)\Big)\|\bm{\lambda}\|_2^2}{\|\bm{A}_d\|_2^2}\cdot\frac{\log(n/\delta)}{\bar{\varepsilon}^2}\Bigg)
\end{equation}
is a sufficient number of queries.
\end{lemma}
\noindent This returns our prediction for one-shot estimation if $V = I$. In this case we find $\sigma_{\min}(V\odot V) = \sigma_{\min}(I) = 1$ and the bound \eqref{sigma min bound} is zero, so $s = 1$ is sufficient to recover the diagonal exactly. Further analysis of the value of $\sigma_{\min}(V\odot V)$ is complex and outside the scope of this paper. We point the interested reader to \cite{ando1987singular} for further investigation.

\subsection{Links with Hutchinson's estimator variance}
Lastly, once again let us turn our attention to the quantity:
\begin{equation} \label{Hutchinson's Variance Quantity}
    \|A\|_F^2 - \|\bm{A}_d\|_2^2
\end{equation}
and recall that, but for a factor of 2, this is in fact the variance of Hutchinson's estimator for the trace. A variety of papers suggest particular variance reduction schemes for Hutchinson's estimator, focused on reducing \eqref{Hutchinson's Variance Quantity}. Certain approaches take advantage of the sparsity of the matrix $A$ \cite{stathopoulos2013hierarchical,tang2011domain} (by a process often called \emph{probing}), whilst others still take a ``decomposition" approach \cite{adams2018estimating}. Perhaps the most promising approach in this case is that taken by Gambhir et.\ al \cite{gambhir2017deflation}, Meyer et.\ al \cite{meyer2021hutch++} and Lin \cite{lin2017randomized}. In these papers the approach is to form a decomposition by approximate projection onto the top eigenspace of $A$, using the matrix $Q\in\mathbb{R}^{n\times r}$ designed to capture most of the column space of $A$. This matrix roughly spans the top eigenspace of $A$. In particular, \cite{gambhir2017deflation} and \cite{lin2017randomized} justify this approach for estimating the trace when $A$ is low rank as $\tr(QQ^TA) = \tr(Q^T A Q)$ will capture most of the trace of $A$. A moment's thought suggests that this method may be adapted to estimate the whole diagonal, since $Q Q^T A$ or indeed $Q Q^T A Q Q^T$ are randomised low-rank approximations to the matrix $A$ \cite{halko2011finding}. We might do well to find the exact diagonals of either $QQ^TA$ or $QQ^T AQQ^T$ and then estimate the diagonal of the remainder of the matrix: $(I - QQ^T)A$ or $(I - QQ^T) A (I - QQ^T)$. We expand further on this idea in Section \ref{Improved diagonal estimation} extending the results and analysis of \cite{meyer2021hutch++}.

\section{Numerical Experiments} \label{Numerical Experiments on the diagonal}
\subsection{Single element estimates}
We first examine the Rademacher diagonal estimator and explore the implications of equations \eqref{single element rademacher probability} and \eqref{single element rademacher queries}, since the sufficient query bound for such an estimator is likely the tightest derived thus far. Moreover, the value of $\varepsilon$ in Theorem \ref{Rademacher Theorem} is arbitrary, so any theoretical bound ought to hold for all values $s$, at all values of $\delta$. For different diagonal elements we expect the same rates of convergence, but different scaling, according to the value of $(\|\bm{A}\|_2^2 - A_{ii}^2)/A_{ii}^2$ for each $i$.

 This is exactly what is found in practice, as shown in Figure \ref{fig:Single Diagonal Estimate Rademacher}. This figure shows the relative convergence of estimates of selected diagonal elements of the real world matrix ``Boeing msc10480" from the SuiteSparse matrix collection\footnote{Formerly the University of Florida matrix collection.} \cite{davis2011university}. These elements were chosen to be displayed as they have large differences in their values of $(\|\bm{A}\|_2^2 - A_{ii}^2)/A_{ii}^2$, and so are most easily resolved in the plots of relative error. As expected, the larger the value of $(\|\bm{A}\|_2^2 - A_{ii}^2)/A_{ii}^2$, the greater the error after a given number of queries.

 For a fixed value of $\delta$ and a given matrix entry, the ratio of the numerical error and theoretical bound is roughly fixed - the convergence of the numerical experiments mirrors that of the theoretical convergence bound. Thus the theoretical bound is tight in $\varepsilon$, as numerical error also scales with $O(1/\varepsilon^2)$. In addition, for decreasing $\delta$, the bound to observation ratio can be seen to decrease, suggesting that the query bound is tight in $\delta$, since we would have otherwise expected it to grow. 

\begin{figure}[!htb]
    \centering\begin{tabular}{cc}
        \hspace{-1.2cm}\includegraphics[width=60mm]{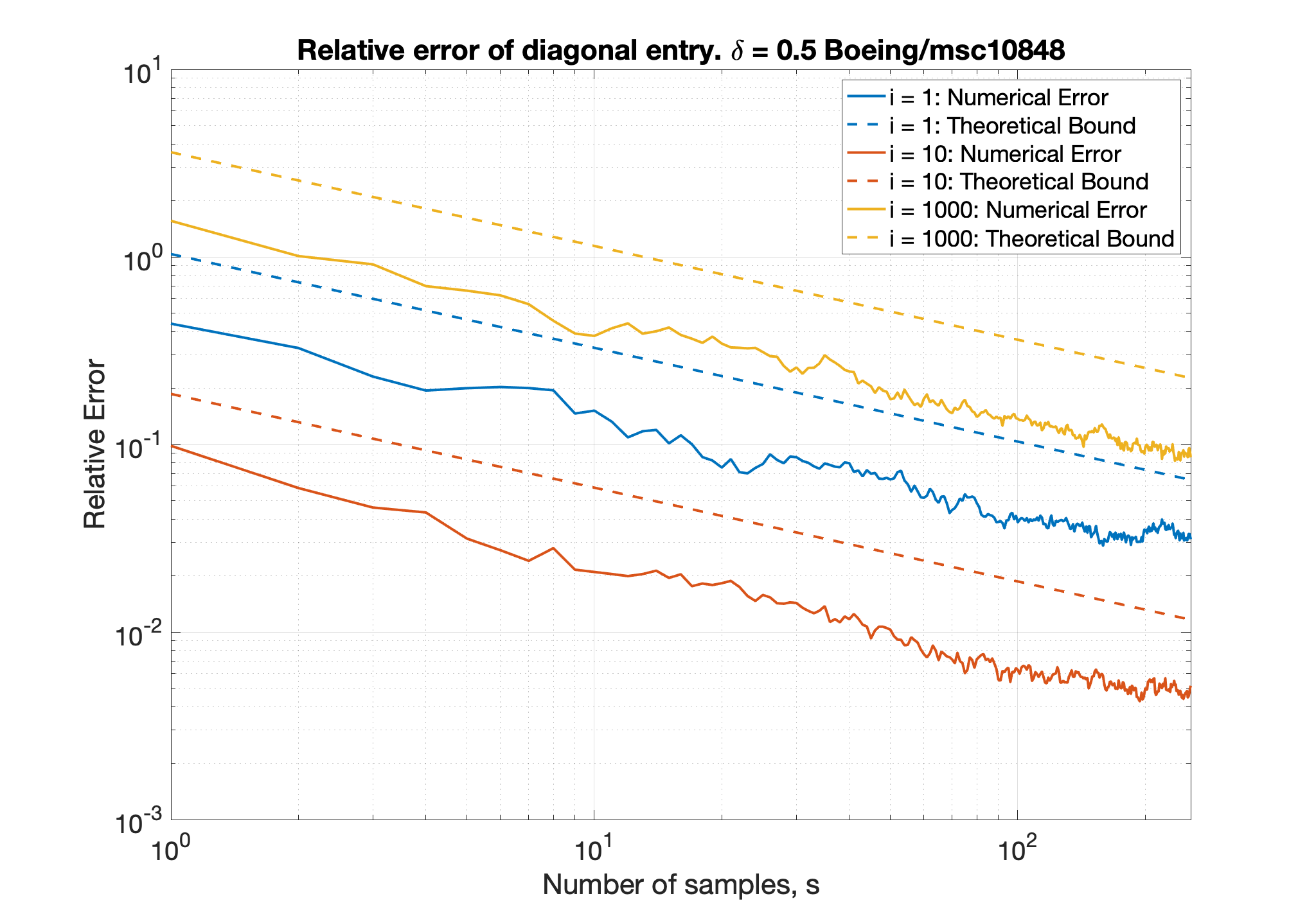} &  \includegraphics[width=60mm]{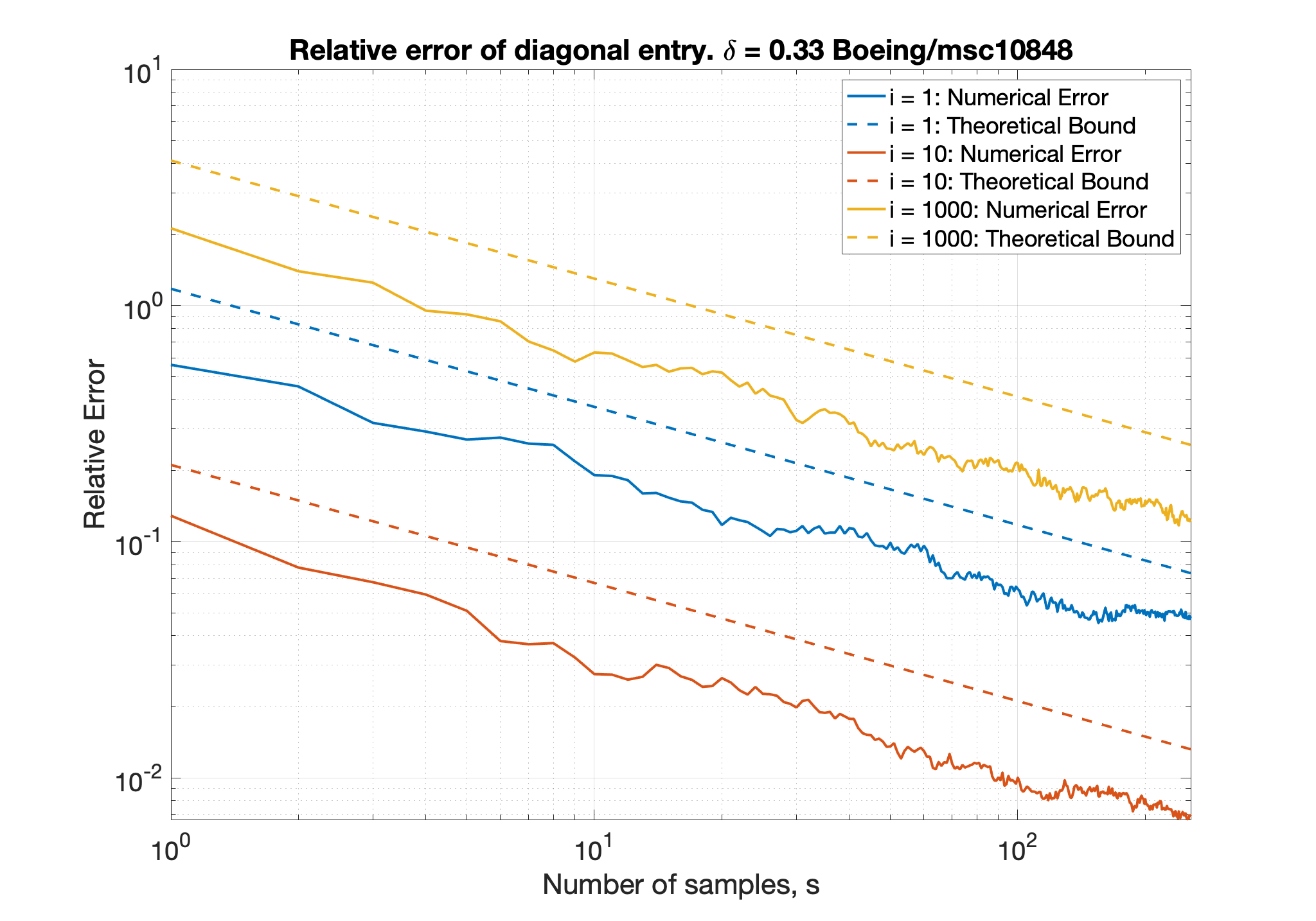} \\
        \hspace{-1.2cm}\includegraphics[width=60mm]{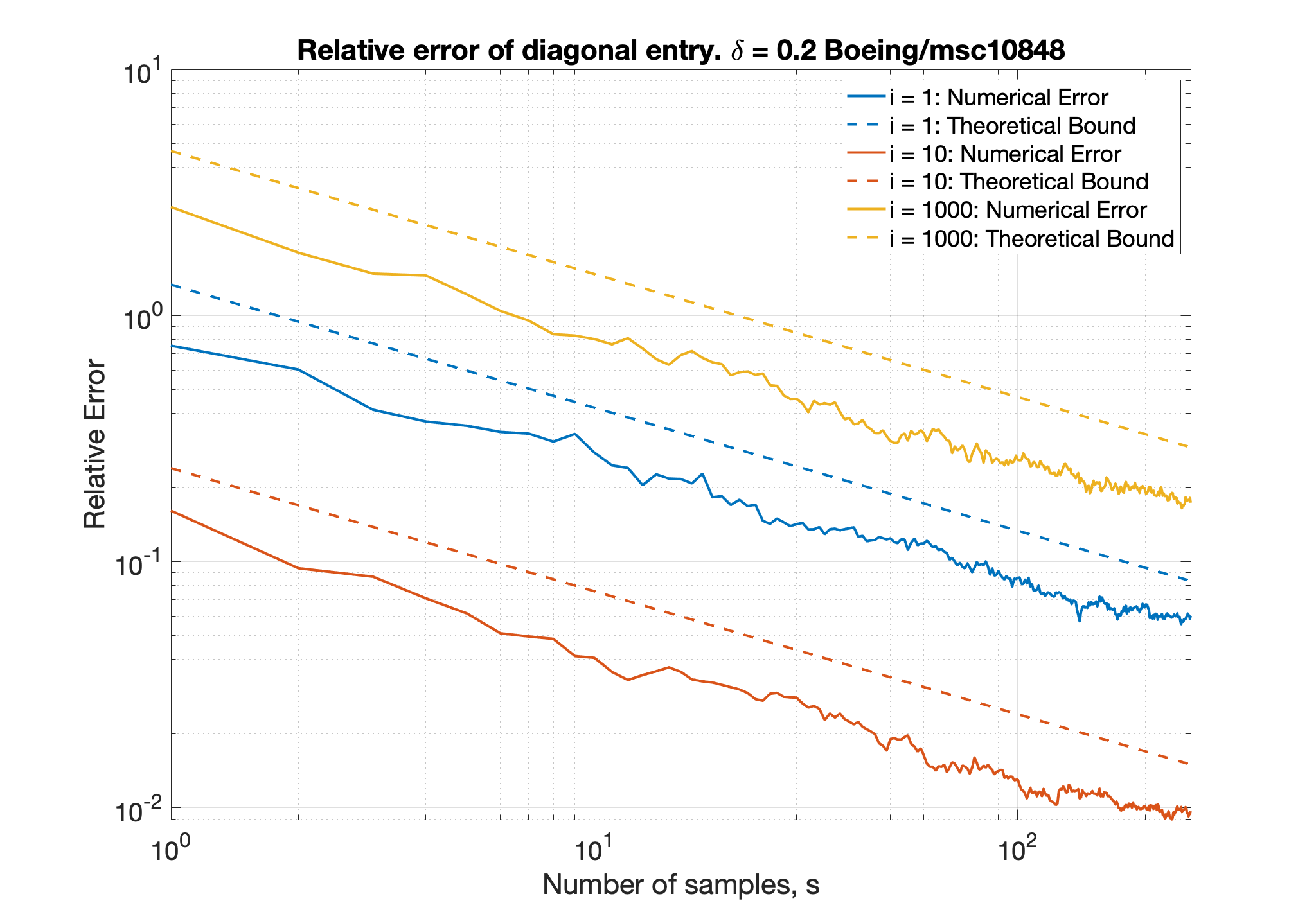}  & \includegraphics[width=60mm]{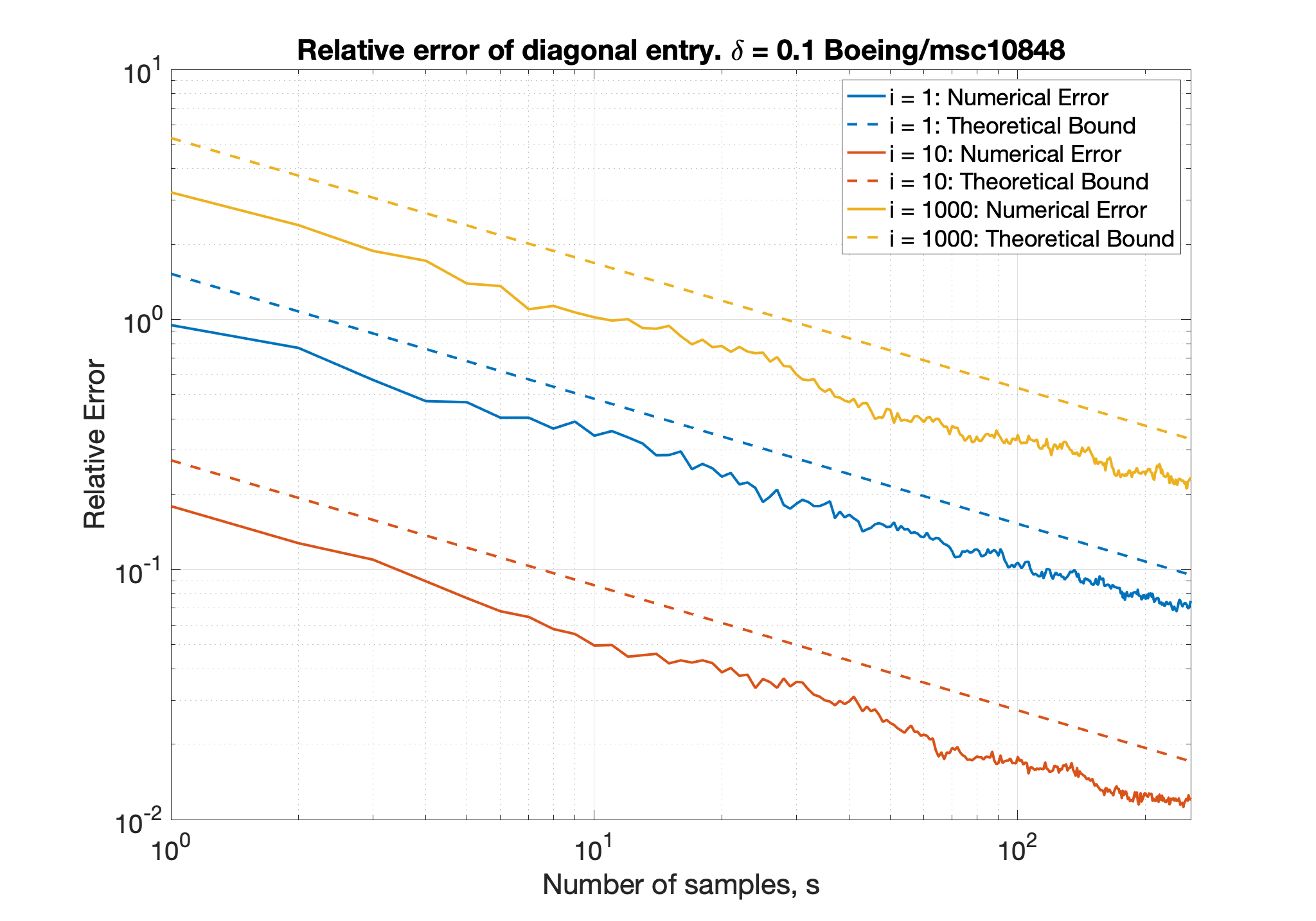} 
    \end{tabular}
    
    \caption{\textit{The convergence of Rademacher estimates of selected diagonal elements of the matrix ``Boeing/msc10480" from the SuiteSparse matrix collection. The numerical errors are shown in solid, and the theoretical bounds in dash. For $100$ trials, we plot the median error (top left), the $67^{th}$ percentile error (top right), the $80^{th}$ percentile error (bottom left) and the $90^{th}$ percentile error (bottom right), along with the corresponding theoretical error bounds for the associated value of $\delta$. The constants $(\|\bm{A}_i\|_2^2 - A_{ii}^2)/A_{ii}^2$ associated with each index are found as $0.3868$, $0.0125$ and $4.7154$ for $i$ equal to $1$, $10$ and $1000$ respectively.}}
    \label{fig:Single Diagonal Estimate Rademacher}
\end{figure}

\subsection{Rademacher vs Gaussian element estimates}
\begin{table}[h!]
    \centering
    \caption{\textit{$^*$Sufficient query bounds for a row-dependent $(\varepsilon, \delta)$-approximator as found by Theorems \ref{Gaussian Theorem} and \ref{Rademacher Theorem}, with the corresponding demands on $\varepsilon$.}}
    \label{tab:Diag comparison}
    \vspace{0.2cm}
    \begin{tabular}{c|cc}
         Query vector distribution & Rademacher &  Gaussian\\\hline
         Sufficient query bound$^*$ &$2\ln(2/\delta)/\varepsilon^2$& $4\log_2(\sqrt{2}/\delta)/\varepsilon^2$\\\hline
         Demand on $\varepsilon$& none & $\varepsilon \in (0,1]$
    \end{tabular}
%    \vspace{0.2cm}
\end{table}

\begin{figure}[!htb]
    \centering
    \begin{tabular}{cc}
         \hspace{-1.2cm}\includegraphics[width=60mm]{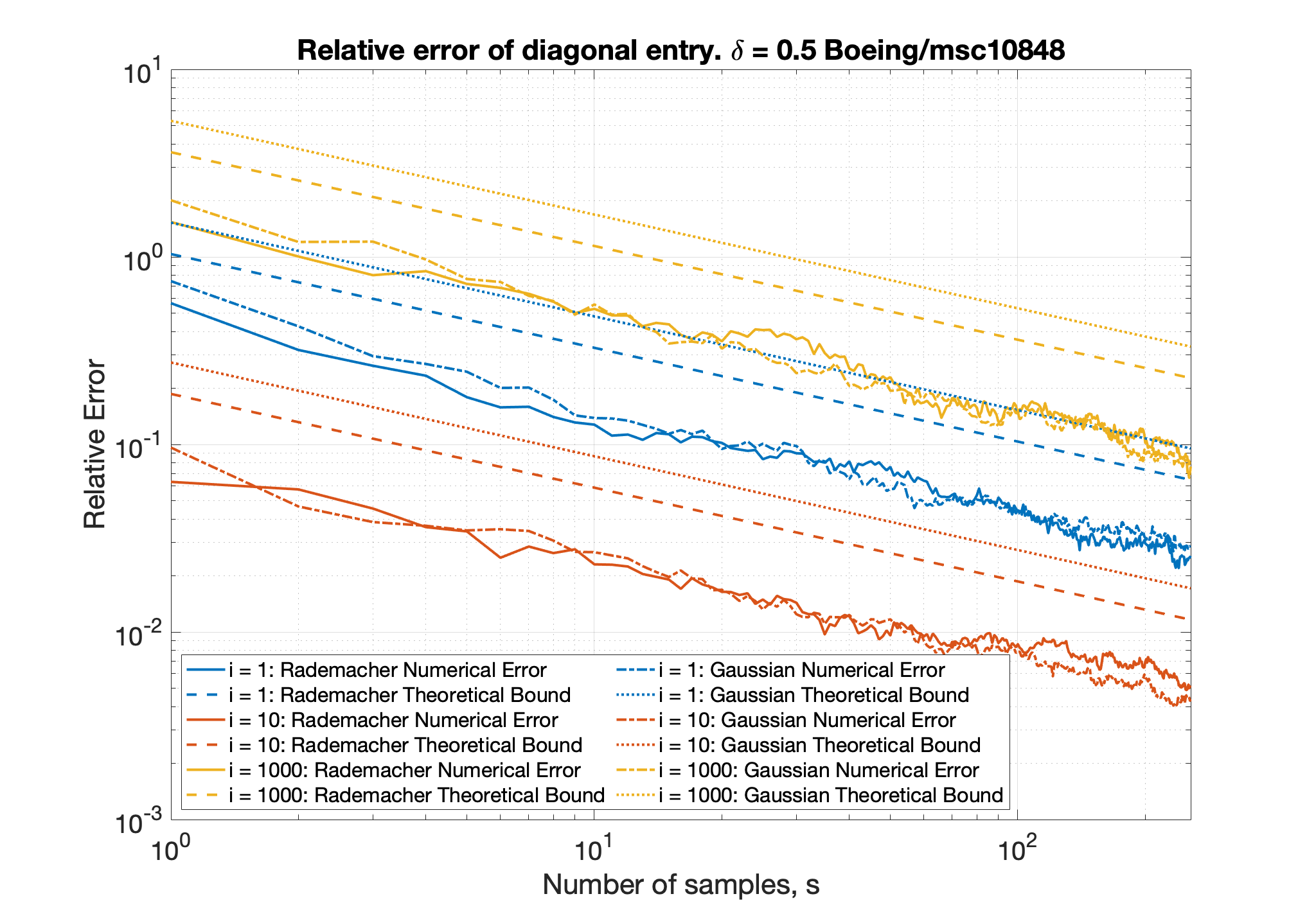}&
         \includegraphics[width=60mm]{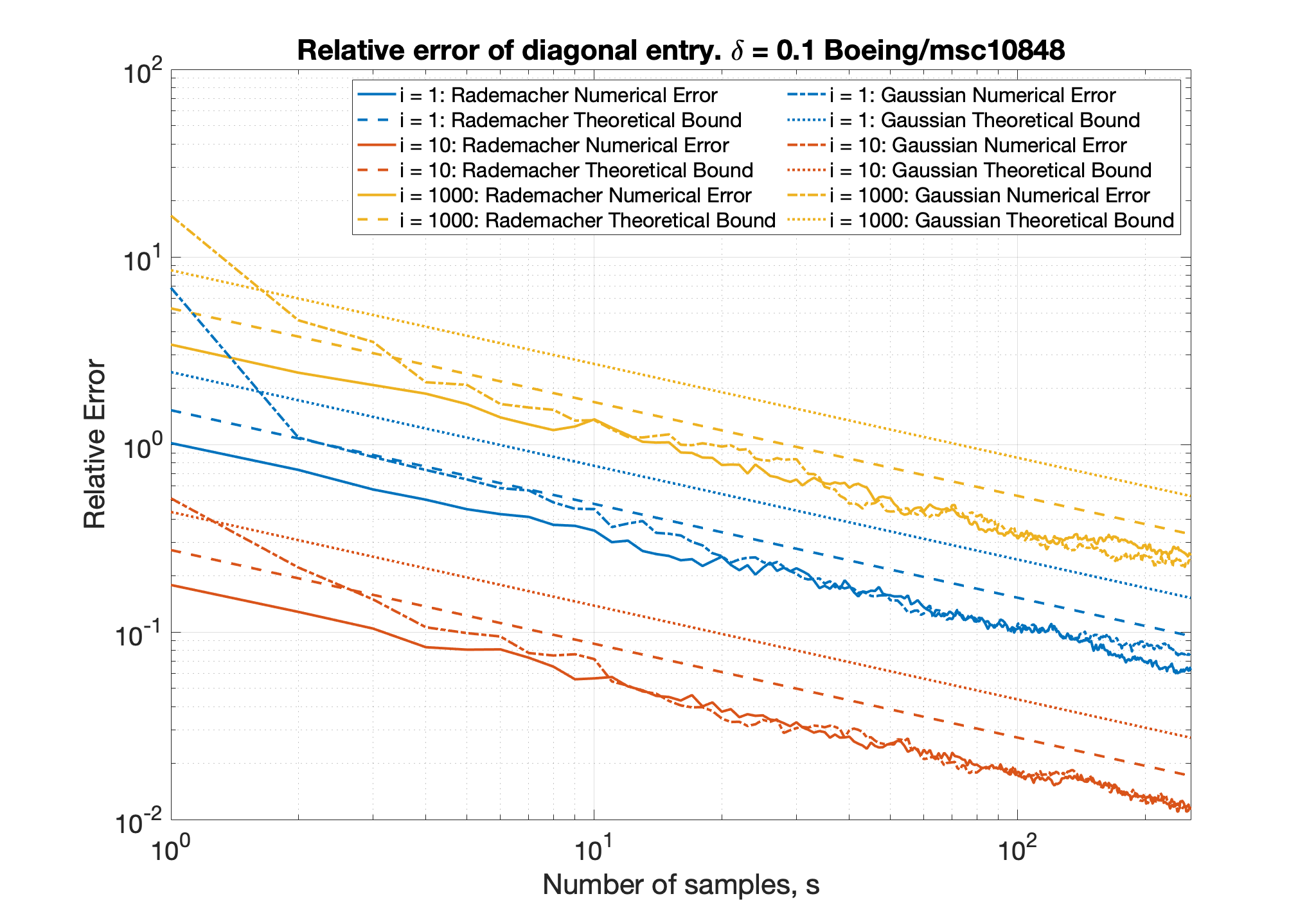}
    \end{tabular}
    \caption{\textit{Comparison of the convergence of the Rademacher and Gaussian diagonal estimators. We run 100 trials for each distribution and report the median error and the $90^{th}$ percentile error in each case, along with their respective theoretical bounds.}}
    \label{fig:Rademacher vs Gaussian estimates figure}
\end{figure}

Table \ref{tab:Diag comparison} summaries sufficient query bounds for a row-dependent $(\varepsilon,\delta)$-approximator for ease of reference. For a relative $(\varepsilon,\delta)$-approximator, we simply introduce the matrix constant $(\|\bm{A}_i\|_2^2 - A_{ii}^2)/A_{ii}^2$ in each bound. The bound for the Rademacher diagonal estimator is lower than that for the Gaussian diagonal estimator, so we expect to be able to reach a lower error faster with the Rademacher estimator. Note however, that since we have not tuned the proof of Theorem \ref{Gaussian Theorem} for optimal $\lambda$ (see equation \eqref{Gaussian Proof End}), the query bound might indeed be lower, but not by much. Experiments suggest that \eqref{Gaussian Proof End} is quite tight unless $s\gg \frac{1}{\epsilon^2}$.

 Figure \ref{fig:Rademacher vs Gaussian estimates figure} illustrates this comparison, again on the matrix ``Boeing/msc10480". For the median experiments, $\delta = 0.5$, the Gaussian theoretical bound contains the corresponding experimental error, but for the $90^{th}$ percentile it seems as if this is not the case. In this instance, where $\delta = 0.1$, the lowest value $s$ may take is when $\varepsilon = 1$. Hence
\begin{equation}
    s > 4\log_2(\sqrt{2}/0.1)/1^2 = 15.3
\end{equation}
so we can only expect the theoretical bound to be valid for $s\geq 16$. If only one Gaussian query is made the relative error appears to not be bounded by the theoretical limit, but this is no surprise, since Theorem \ref{Gaussian Theorem} poses the question: for given $\delta < 1$ and $\varepsilon \in (0,1]$, what is $s$ bounded by? Whereas experimentation shows us, for given $s$ and $\delta$, the value $\varepsilon$ will take\footnote{Note that plots are of relative error, but $\varepsilon$ in the context of Theorem \ref{Gaussian Theorem} refers to row-dependent $(\varepsilon, \delta)$-approximators.}. The theorem gives us a tool to find $s$, but in return asks that $\varepsilon$ is bounded\footnote{For equation \eqref{F bound} to be valid.} from above by $1$. The experiment, meanwhile, takes an input of $s$ and $\delta$ and returns $\varepsilon$.

 Whilst the theoretical bound in the right hand plot of Figure \ref{fig:Rademacher vs Gaussian estimates figure} only kicks in at $s \geq 16$, we have drawn it extending back to $s = 1$ to illustrate the fact that it does not apply for $s$ small enough. Meanwhile, as noted, the Rademacher bound holds for arbitrary $\varepsilon$. The effect is further manifested by the fact that with $\delta = 0.1$, the Rademacher estimator is more accurate for sufficiently small $s$, roughly $s \lesssim 10$.

\subsubsection{Intuitions for performance with small $s$}
\noindent To gain more intuition for why the Rademacher diagonal estimator is better than the Gaussian diagonal estimator for small $s$, it is useful to compare the variances of individual elements for the two estimators. From equation \eqref{Single element variance Rademacher}, the variance of the Radamacher estimate is
\begin{equation}
    \Var(R_i^s) = \frac{1}{s}\sum_{j\neq i}A_{ij}^2,
\end{equation}
whilst in the Gaussian case for $s>2$, we find
\begin{equation} \label{Gaussian Variance}
    \Var(G_i^s) = \frac{1}{s - 2}\sum_{j\neq i} A_{ij}^2
\end{equation}
but for $s \leq 2$, the variance is not bounded. For a given $i$, this is found in the following manner
\begin{equation*}
    \begin{split}
        \Var(G_i^s) &= \Var(G_i^s - A_{ii})\\
        &= \Var\Bigg(\mathlarger{\sum_{j \neq i}^n} A_{ij} \Bigg(\frac{\sum_{k=1}^s v_k^i v_k^j}{\sum_{k=1}^s{(v_k^i)^2}}\Bigg)\Bigg)\\
        & = \Var\Bigg(\mathlarger{\sum_{j \neq i}^n} A_{ij}\frac{\bm{u}_i^T\bm{u}_j}{\|\bm{u_i}\|_2^2}\Bigg)\\
        &=\mathlarger{\sum_{j \neq i}^n} A_{ij}^2\Var\Bigg(\frac{\bm{u}_i^T\bm{u}_j}{\|\bm{u_i}\|_2^2}\Bigg)\hspace{1cm}\text{(independence wrt $j$)}
    \end{split}
\end{equation*}
where $\bm{u}_i$ and $\bm{u}_j$ are as in the proof of Theorem \ref{Gaussian Theorem}. Referring back to this proof again, we already know that the distribution of $(\bm{u}_i^T\bm{u}_j/\|\bm{u}_i\|_2^2)$ is the same as $g_j/\|\bm{u}_i\|_2$ for some $g_j \sim N(0,1)$ independent of $\|\bm{u}_i\|_2$. This gives
\begin{equation*}
    \begin{split}
        \Var(G_i^s) &= \sum_{i=1}^n A_{ij}^2\Bigg(\mathbb{E}\Bigg[\frac{g_j^2}{\|\bm{u}_i\|_2^2}\Bigg] - \mathbb{E}\Bigg[\frac{g_j}{\|\bm{u}_i\|_2}\Bigg]^2\Bigg)\\
        &= \sum_{i=1}^n A_{ij}^2\mathbb{E}\Bigg[\frac{g_j^2}{\|\bm{u}_i\|_2^2}\Bigg]\\
        &=\sum_{i=1}^n A_{ij}^2 \mathbb{E}\Bigg[\frac{\chi_1^2}{\chi_s^2}\Bigg]\\
        &= \sum_{i=1}^n A_{ij}^2 \mathbb{E}\Bigg[\frac{1}{s}F(1,s)\Bigg]\\
    \end{split}
\end{equation*}
where $F(1,s)$ is an $F$-distribution with parameters $1$ and $s$. The expectation of such a variable is readily found as $s/(s-2)$, for $s > 2$, and otherwise is infinite, yielding the result of equation \eqref{Gaussian Variance}. Thus we expect the estimators to be similarly behaved for sufficiently large $s$ but not if $s$ is small, in which case we expect that the Rademacher estimator shall be the better of the two, as the numerical results indicate.

\subsubsection{Implications for related work}
\noindent This has implications for the recent paper of Yao et.\ al \cite{yao2020adahessian}. The authors estimate the diagonal of the Hessian of a loss function in machine learning tasks, implicit from back-propagation. Only a single Rademacher query vector is used due to computational restrictions. These results suggest it would be unwise to switch to a single Gaussian query vector as in Definition \ref{Gaussian definition}, since such an estimate would have unbounded variance.\\

\noindent We remark further that our the results, particularly Theorem \ref{Rademacher Theorem} and Corollary \ref{Rad Union Theorem}, suggest the algorithm may greatly benefit from taking further estimates of the diagonal. Indeed the authors indicate \cite{yao2020video} that  the noise the algorithm experiences is dominated by the error of estimation of the diagonal. It is now clear why.

\subsection{Full diagonal estimates}
Figure \ref{fig:Full Diagonal 2-norm Radmacher} shows the convergence of the Rademacher and Gaussian diagonal estimators to the full diagonal. Note that, in contrast to the experiments for single elements of the diagonal, the theoretical bounds appear looser. This is expected, since in arriving at the theoretical bounds for our complete diagonal estimates we have used multiple inequalities for single elements. Therefore the query bounds will be loose. Use of a union bound in Corollaries \ref{Gaussian Union Theorem} and \ref{Rad Union Theorem} shall further compound this loosening. Note once again that, in the Gaussian case, the theoretical bound holds only for sufficiently large $s$, as expected. The effects seen in the error of single diagonal element estimates carry over to the full diagonal, but are even more pronounced. For a handful of queries (again, $s \lesssim 10$ or so) the Rademacher estimator out performs its Gaussian counterpart, but as a greater number of queries are used the convergence behaviours become increasing similar. As before, with the Rademacher estimator the ratio of the theoretical bound to the experimental error is roughly constant for each $\delta$, so the convergence is of the order of $O(1/\varepsilon^2)$, just as for single elements. In the Gaussian case, for large enough $s$, the same may be said.

\begin{figure}
    \centering\begin{tabular}{cc}
        \hspace{-1.2cm}\includegraphics[width=60mm]{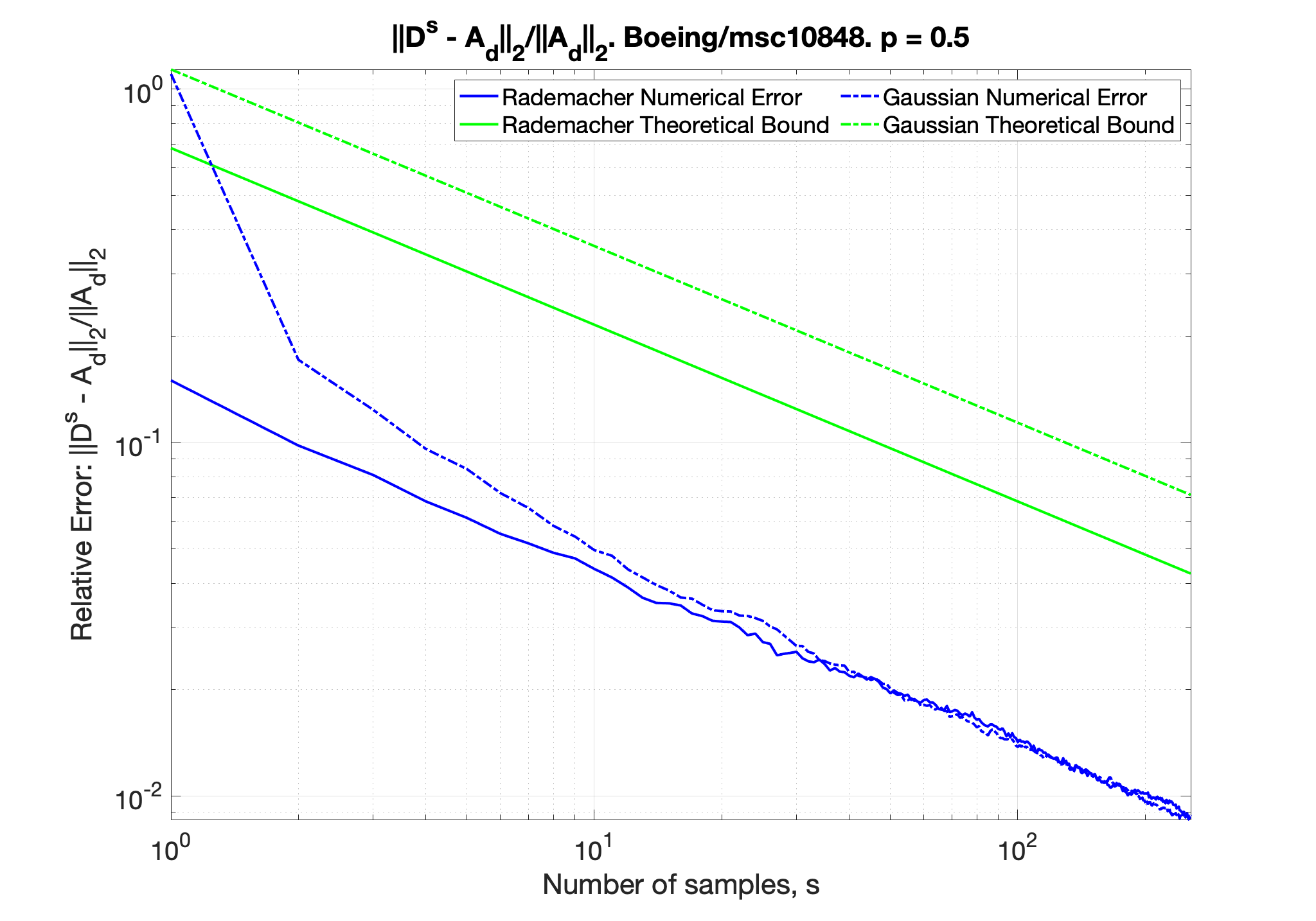}
        \includegraphics[width=60mm]{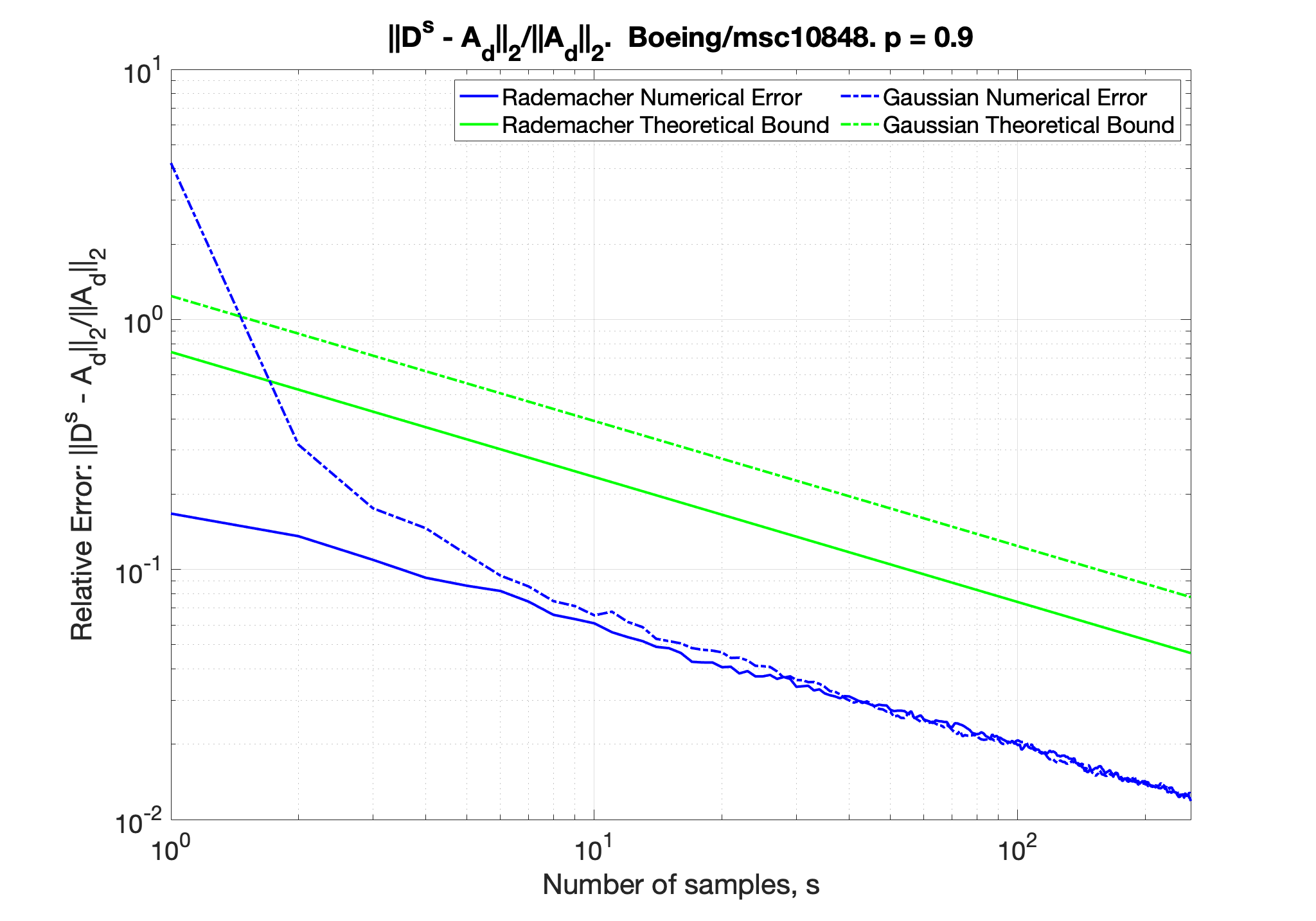}
    \end{tabular}
    \caption{\textit{Convergence of the Rademacher and Gaussian diagonal estimators to the actual diagonal of the matrix ``Boeing/msc10480" from the SuiteSparse matrix collection. The effect of the infinite variance for the Gaussian estimator when $s\leq2$ is much more pronounced.}}
    \label{fig:Full Diagonal 2-norm Radmacher}
\end{figure}

\begin{figure}
    \centering\begin{tabular}{cc}
    \hspace{-1.2cm}
    \includegraphics[width=60mm]{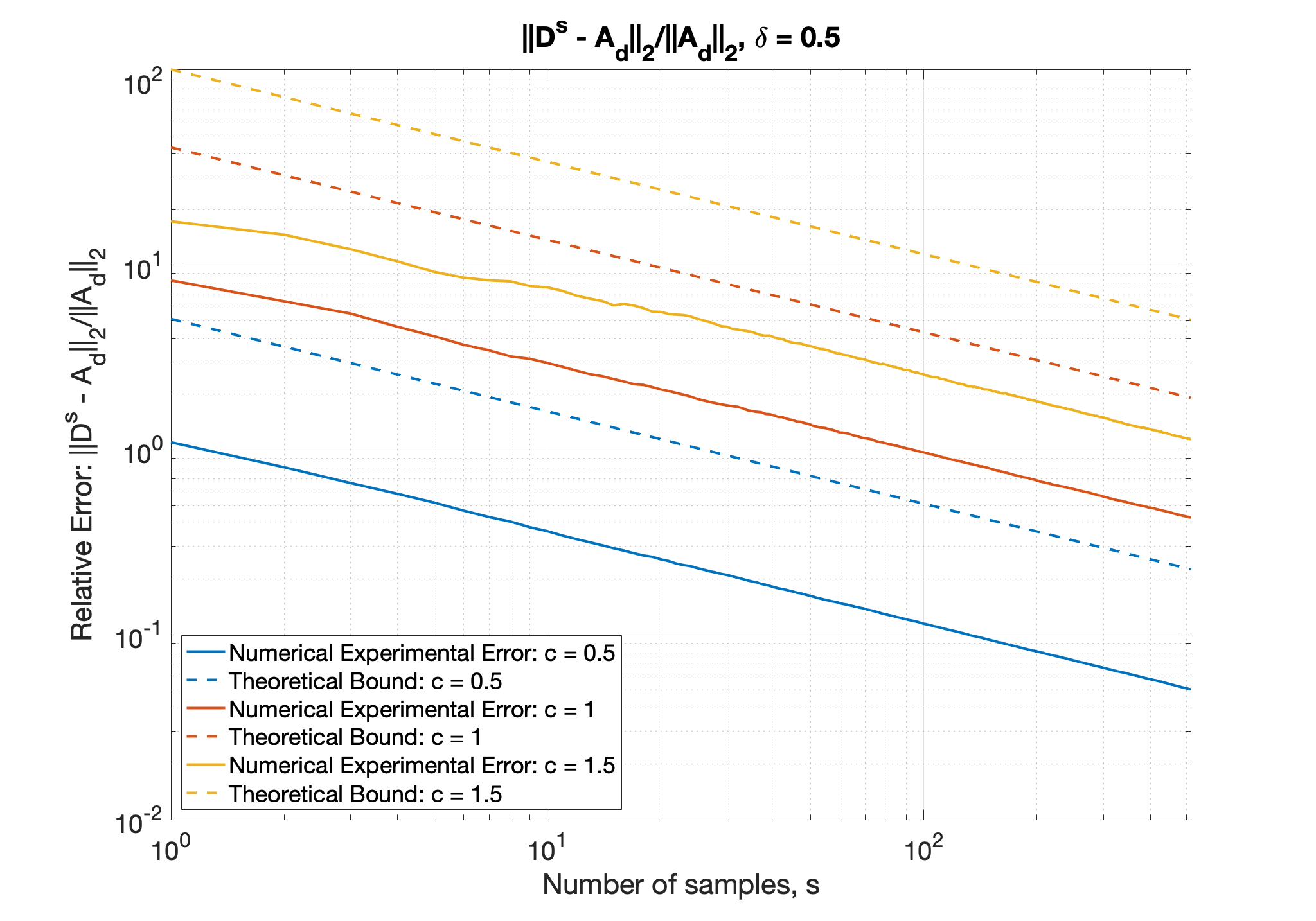}
    \includegraphics[width=60mm]{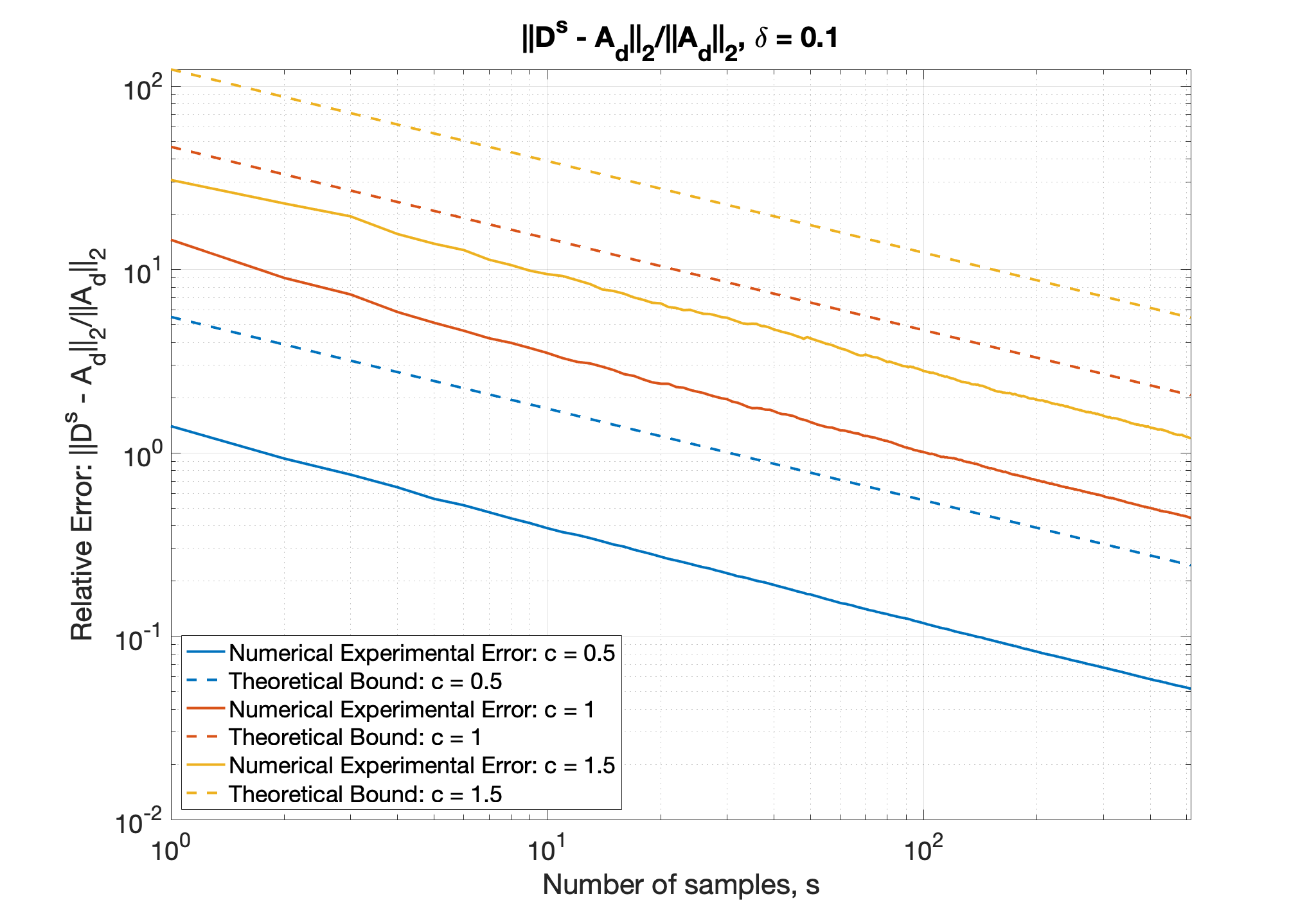}
    \end{tabular}
    \caption{\textit{Convergence of the Rademacher estimator on random matrices. We run $50$ trials and report the median and $90^{th}$ percentile errors for three values of $c,$ $0.5$, $1$ and $1.5$, corresponding to slow, medium and fast eigenvalue decay respectively. As expected, with all matrices, convergence goes as $O(1/\varepsilon^2)$. The difference between the matrices arises from the fact that the values of $(\|A\|_F^2 - \|\bm{A}_d\|_2^2)/\|\bm{A}_d\|_2^2$ increase as $c$ increases.}}
    \label{fig:Synetic Matrix figure}
\end{figure}

\subsection{Experiments on Synthetic Matrices.}
We also test the Rademacher estimator on random matrices with power law spectra. We take $\Lambda$ to be diagonal with $\Lambda_{ii} = i^{-c}$, for varying c. We generate a random orthogonal matrix $V\in \mathbb{R}^{5000\times 5000}$ by orthogonalising a Gaussian matrix, and set $A = V^T \Lambda V$. A larger value of $c$ results in a more quickly decaying spectrum, meaning that the quantity $\|\bm{\lambda}\|_2^2 - \frac{1}{n}\|\bm{\lambda}\|_1^2$ (as in equation \eqref{lambda2 vs lambda1}) will be larger, so we expect worse estimates for a given number of queries. This is exactly as found empirically, as shown in figure \ref{fig:Synetic Matrix figure}. In particular, for $c = 1$ and $c = 1.5$, the error of our estimates is still woeful even after $s = 512$ vectors! This does not bode well for stochastic diagonal estimation of general dense matrices with sharply decaying spectra. Whilst we still have $O(1/\varepsilon^2)$ convergence, the error starts off too poorly for this to be of much use.

\subsection{Summary}
\noindent At the start of this section, we set out to find an answer for the number of queries needed for a given error estimation. We have seen that, for fixed $\delta$, $s$ goes as $O(1/\varepsilon^2)$, and that the error scales with $(\|\bm{A}_i\|_2^2 - A_{ii}^2)^{1/2}$. We have seen that the Rademacher diagonal estimator has both a better theoretical bound than its Gaussian counterpart and also has lower variance for each element. This is in contrast to the conclusions of Avron and Toledo \cite{avron2011randomized}, who found that Gaussian trace estimators were better than Hutchinson's, Rademacher estimator. However, there is no real contention here. Summing the components of the Gaussian diagonal estimator is not the same as the Gaussian estimator of the trace examined in \cite{avron2011randomized}. The difference is merely an artifact of the denominator introduced in \eqref{Gaussian definition equation}.\\

\noindent We have also seen that positive semi-definite matrices with flatter spectra are suitable for estimation in this way, but as the spectrum of a matrix becomes steeper, estimates become worse. Perhaps this may spell the end of stochastic diagonal estimation. Perhaps it may restrict it to a specific class of matrices, in contexts where the spectrum of the matrix is known to be close to flat due to properties of the system under investigation. All things considered, it so far seems that, at the best, this form of estimation may only be useful under special circumstances. Might there be a way to improve on this poor performance for general, dense matrices, with rapidly decaying spectra? This is the subject of Section \ref{Improved diagonal estimation}. Since it is found to be the best, we take forward only the Rademacher diagonal estimator.

\section{Improved diagonal estimation}\label{Improved diagonal estimation}
Is it possible to improve stochastic diagonal estimation, particularly the worst case scenarios of the previous section? This section seeks to address this by extending the ideas of \cite{meyer2021hutch++}, adapting the algorithm therein for full diagonal estimation. For $A\succeq 0$, we propose an $O(1/\varepsilon^2)$ to $O(1/\varepsilon)$ improvement in the number of required query vectors and illustrate the robustness of our algorithm to matrix spectra.

\subsection{Relation to Hutch++} \label{Relation to Hutch++}
\noindent We briefly give an overview of the work of Meyer, Musco, Musco and Woodruff \cite{meyer2021hutch++}. In their paper, they motivate their algorithm, ``Hutch++" as a quadratic improvement in trace estimation. Using randomised projection to approximately project off the top eigenvalues, they then approximate the trace of the remainder of the matrix. All of the error results from this remainder.

 They arrive at their result in the following manner, first introducing the following lemma (Lemma 2, Section 2, \cite{meyer2021hutch++}) for a general square matrix $A$.
\begin{lemma} \label{Lemma 2 Hutch++}
    For $A \in \mathbb{R}^{n\times n},\; \delta \in (0, 1/2],\; s \in \mathbb{N}$. Let $\textnormal{tr}_R^s(A)$ be the Hutchinson estimate of the trace. For fixed constants $c, C$, if $s \geq c\log(1/\delta)$, then with probability $1-\delta$
    \begin{equation*}
        |\textnormal{tr}_R^s(A) - \textnormal{tr}(A)| \leq C \sqrt{\frac{\log(1/\delta)}{s}} \|A\|_F.
    \end{equation*}
Thus, if $s = O(\log(1/\delta)/\varepsilon^2)$ then, with probability $1-\delta$, $|\textnormal{tr}_R^s(A) - \textnormal{tr}(A)| \leq \varepsilon\|A\|_F$.
\end{lemma}

\noindent So, for $A \succeq 0$, we simply have $\|A\|_F = \|\bm{\lambda}\|_2 \leq \|\bm{\lambda}\|_1 = \tr(A)$, to get a relative trace approximator. Note, this bound is only tight when $\|\bm{\lambda}\|_2 \approx \|\bm{\lambda}\|_1$. The authors then give the following lemma (see Lemma 3, Section 3, \cite{meyer2021hutch++}).
\begin{lemma} \label{Lemma 3 Hutch++}
    Let $A_r$ be the best rank-r approximation to PSD matrix $A$. Then
    \begin{equation*}
        \|A-A_r\|_F \leq \frac{1}{\sqrt{r}}\textnormal{tr}(A).
    \end{equation*}
\end{lemma}
\begin{proof}
Since $\lambda_{r+1} \leq \frac{1}{r}\sum_{i=1}^r\lambda_{i} \leq \frac{1}{r}\tr(A)$,
\begin{equation*}
    \|A - A_r\|_F^2 = \sum_{i = r+1}^n \lambda_i^2 \leq \lambda_{r+1}\sum_{i=k+1}^n\lambda_{i}\leq\frac{1}{r}\tr(A)\sum_{i=r+1}^n\lambda_{i} \leq \frac{1}{r}\tr(A)^2.
\end{equation*}
\hfill$\square$
\end{proof}

\noindent They argue that this result immediately lends itself to the possibility of an algorithm with $O(1/\varepsilon)$ complexity. Setting $s = r = O(1/\varepsilon)$ and splitting $\tr(A) = \tr(A_s) + \tr(A - A_s)$, the first term may be computed exactly if $V_s\in \mathbb{R}^{n\times s}$ (the top $s$ eigenvectors) is known, since $\tr(A_s) = \tr(V_s^T A V_s)$. Thus the error in the estimate of the trace is contained in our estimate of $\tr(A - A_s)$. So (roughly) combining Lemmas \ref{Lemma 2 Hutch++} and \ref{Lemma 3 Hutch++}, the error becomes, for their trace estimator: $\tr_{H++}^s(A)$
\begin{equation} \label{Lemmas 2 and 3 synthesis}
    \begin{split}
        |\tr_{H++}^s(A) - \tr(A)| &\approx |\tr_R^s(A - A_s) -\tr(A-A_s)|\\
        &\leq  C \sqrt{\frac{\log(1/\delta)}{s}}\|A - A_s\|_F\\ &\leq  C \frac{\sqrt{\log(1/\delta)}}{s}\tr(A).
    \end{split}
\end{equation}\\
Hence $s = O(1/\varepsilon)$ for fixed $\delta$ is sufficient for an $(\varepsilon,\delta)$-approximator to the trace of $A$.

 Of course, $V_s$ cannot be computed exactly in a handful of matrix-vector queries, but this is easily resolved with standard tools from randomised linear algebra \cite{halko2011finding}. Namely, $O(s)$ queries is sufficient to find a $Q$ with $\|(I-QQ^T)A(I-QQ^T)\|_F \leq O(\|A-A_s\|_F)$ with high probability, which is all that is required in the second line of \eqref{Lemmas 2 and 3 synthesis}, and how the algorithm is implemented in practice. For a full, formal treatment of these ideas see \cite{meyer2021hutch++}, Theorem 4, Section 3. We elaborate on this in the next section.

\subsection{Diag++} \label{Diag++}
Here we outline how the above ideas may be applied to estimation of the entire diagonal, using $O(1/\varepsilon)$ queries. Recall once again, for 
\begin{equation}
    \Pr\Bigg(\|\bm{D}^s - \bm{A}_d\|_2^2 \leq \varepsilon^2\Big(\|A\|_F^2 - \|\bm{A}_d\|_2^2\Big)\Bigg) \geq 1 - \delta
\end{equation}
it is sufficient to take
\begin{equation}
    s > O\Bigg(\frac{\log(n/\delta)}{\varepsilon^2}\Bigg).
\end{equation}
We may now engage in a trade-off. For fixed $\delta$ and $n$, we swap $O(1/\varepsilon^2)$ convergence for $O(1/\varepsilon)$ convergence. This comes at the (potentially minor) expense of a greater multiplicative factor in our query bound. Just as in Hutch++, we first motivate this supposing we know $V_s$, the top $s$ eigenvectors, exactly.

 Supposing we have projected off the $s$-rank approximation $A_s$ and extracted its diagonal, $\text{diag}(A_s)$, we are left to estimate the diagonal of 
\begin{equation*}
    \widehat{A}:= A - A_s.
\end{equation*}
Denoting $\widehat{\bm{D}}^s$ as the diagonal estimate of this, we have 
\begin{equation} \label{exact remainder approx}
    \|\widehat{\bm{D}}^s - \widehat{\bm{A}}_d\|_2^2 \leq \varepsilon^2\Big(\|\widehat{A}\|_F^2 - \|\widehat{\bm{A}}_d\|_2^2\Big)
\end{equation}
with probability, $1-\delta$, if we have used a sufficient number of queries: $s > \frac{2}{\varepsilon^2}\ln(2n/\delta)$ for the Rademacher diagonal estimator. Ignoring the $\|\widehat{\bm{A}}_d\|_2^2$ term, and employing the result of Lemma \ref{Lemma 3 Hutch++}, yields
\begin{equation*}
    \|\widehat{\bm{D}}^s -\widehat{\bm{A}}_d\|_2^2 \leq \frac{\varepsilon^2\tr(A)^2}{s}.
\end{equation*}
Letting
\begin{equation*}
    \varepsilon^2 = s\bar{\varepsilon}^2\Bigg(\frac{\|\bm{A}_d\|_2^2}{\tr(A)^2}\Bigg)
\end{equation*}
gives
\begin{equation}
    \|\widehat{\bm{D}}^s - \widehat{\bm{A}}_d\|_2^2 \leq \bar{\varepsilon}^2\|\bm{A}_d\|_2^2
\end{equation}
still with probability $1-\delta$, if we have used sufficient $s$.
Note that this is not a relative estimation of the diagonal of $\widehat{\bm{A}}$; rather the error is qualified relative to $\|\bm{A}_d\|_2^2$, since this is, overall, the error that we are interested in. Explicitly, for the Rademacher diagonal estimator, this results in sufficiency if
\begin{equation} \label{s-to order 1/eps}
    \begin{split}
        s &> \frac{2}{s}\Bigg(\frac{\tr(A)^2}{\|\bm{A}_d\|_2^2}\Bigg)\frac{\ln(2 n/\delta)}{\bar{\varepsilon}^2}\\
        \implies s& > \sqrt{2}\Bigg(\frac{\tr(A)}{\|\bm{A}_d\|_2}\Bigg)\frac{\sqrt{\ln(2 n/\delta)}}{\bar{\varepsilon}}
    \end{split}
\end{equation}
and we have arrived at $O(1/\varepsilon)$ convergence, unlike $O(1/\varepsilon^2)$ as for all previous estimates! We remark that when \eqref{s-to order 1/eps} is tight, this is a dramatic improvement! In the case where $\|\bm{A}_d\|_1/\|\bm{A}_d\|_2 = O(1)$, the result of \eqref{s-to order 1/eps} is very favorable\footnote{Such a case may be found for the ``Model Hamiltonian" in Section 4 of \cite{bekas2007estimator}.}. Since $\tr(A) = \|\bm{A}_d\|_1 \leq \sqrt{n}\|\bm{A}_d\|_2$, the worst case scenario (when the diagonals take uniform values) yields
\begin{equation} \label{Diag++ worst case}
    s > \frac{\sqrt{2n\ln( 2n/\delta})}{\bar{\varepsilon}}
\end{equation}
and we observe the introduction of a potentially expensive constant, but one which, for moderate $\bar{\varepsilon}$ and $\delta$ will nonetheless give $s \ll n$ for large $n$.

 Once again, we have motivated the above for exact $s$-rank approximation, but, as in Hutch++, the same tools from randomised linear algebra \cite{halko2011finding} mean that we can approximate $V_s$ using $O(s)$ queries. Specifically $O(s)$ queries are sufficient to find $Q$ such that 
\begin{equation} \label{randomised remainder bound}
    \|(I-QQ^T)A(I-QQ^T)\|_F \leq O(\|A-A_s\|_F).
\end{equation}
In detail (see \cite{musco2020projection}, Corollary 7 and Claim 1), provided $O(r + \log(1/\delta))$ vectors have been used to form $Q\in\mathbb{R}^{n\times r}$, then, with probability $1-\delta$
\begin{equation*}
    \|A(I-QQ^T)\|_F^2\leq %\rrr{2\|A-A_r\|_2^2}{
2\|A-A_r\|_F^2
\end{equation*}
and since
\begin{equation*}
    \begin{split}
        \|(I - QQ^T)A(I-QQ^T)\|_F &\leq \|(I - QQ^T)\|_2\|A(I-QQ^T)\|_F\\ &= \|A(I-QQ^T)\|_F
    \end{split}
\end{equation*}
so equation \eqref{randomised remainder bound} holds for $s = O(r + \log(1/\delta))$. Thus, overall it is sufficient to take $s = O\Big(\frac{\tr(A)}{\|\bm{A}_d\|_2}\frac{\sqrt{\log(n/\delta)}}{\bar{\varepsilon}} + \log(1/\delta)\Big)$ queries where clearly the first term dominates. Applying the same ideas as in equations \eqref{exact remainder approx} to \eqref{s-to order 1/eps}, for the randomised Rademacher case we readily arrive at 
\begin{equation} \label{Random Diag++ bound}
    s > 4\Bigg(\frac{\tr(A)}{\|\bm{A}_d\|_2}\Bigg)\frac{\sqrt{\ln(2 n/\delta)}}{\bar{\varepsilon}} + c \log(1/\delta)
\end{equation}
for some constant $c$, is sufficient to approximate the diagonal to within $\bar{\varepsilon}$ error, with probability $1-\delta$. Of course, we remark that this might be a very pessimistic bound. We have ignored $\|\widehat{\bm{A}}_d\|_2^2$ and employed the potentially loose result of Lemma \ref{Lemma 3 Hutch++}, so cannot necessarily expect this bound to be good. However, if we have a steep spectrum\footnote{Such an example may be found for the ``Graph Estrada Index" in Section 5.2 of \cite{meyer2021hutch++}.} with both $\|\widehat{A}\|_F^2 \gg \|\widehat{\bm{A}}_d\|_2^2$ and $\|A\|_F \approx \tr(A)$, then the bound on $s$ should result in Diag++ outperforming simple stochastic estimation. Concretely the algorithm for Diag++ with Rademacher query vectors is as follows:
\begin{algorithm}[H]
    \caption{Diag++. Adaption of Hutch++ for diagonal estimation.}
    \textbf{Input:}\text{ Implicit-matrix oracle, for $A \succeq 0$. Number of queries $s$.}\\
    \textbf{Output:}\text{ Diag++($A$): Approximation to diag($A$)}
    \begin{algorithmic}[1]
        \State Draw $R_1\in\mathbb{R}^{n\times s/3}$ with Rademacher entries.
        \State Compute orthonormal basis $Q$ for $AR_1$
        \State Estimate diag($(I-QQ^T)A(I-QQ^T)$) $\approx \widehat{\bm{D}}^{s/3}\newline = \frac{3}{s}\sum_{k=1}^{s/3}\bm{v}_k\odot (I-QQ^T)A(I-QQ^T)\bm{v}_k$\newline for $\bm{v}_k$ with Rademacher entries
        \State \Return Diag++($A$) $= \text{diag}(QQ^TAQQ^T) + \widehat{\bm{D}}^{s/3}$
    \end{algorithmic}
\end{algorithm}
\noindent In total $s$ matrix-vector multiplications with $A$ are required\footnote{A recent preprint~\cite{persson2021improved} suggests that one can often improve the estimate by adaptively choosing the number of queries used for computing $AR_1$, $AQ$ and $\widehat{\bm{D}}^{s/3}$, instead of the equal queries $s/3$ used in Hutch++. It would be of interest to explore such improvement for diagonal estimation.}: $s/3$ for $AR_1$, $s/3$ for $\widehat{\bm{D}}^{s/3}$ and $s/3$ for $AQ$. Computing $Q$ requires $O(s^2 n)$ extra flops, but this is dominated by the cost of $s$ matrix-vector multiplications: $O(s n^2)$.

\subsection{Numerical Experiments} \label{Numerical Experiments Diag++}

To investigate the performance of Diag++, we once again turn to synthetic matrices: $A = V^T\Lambda V$: $\Lambda_{ii} = i^{-c}$ and $V$ a random orthogonal matrix $V\in \mathbb{R}^{5000\times 5000}$ formed by orthogonalising a Gaussian matrix. Recall that with stochastic estimation alone, convergence was slow for such matrices with steep spectra. We expect Diag++ to perform more or less similarly to standard stochastic estimation if the spectrum of the matrix is flat, since \eqref{Random Diag++ bound} is pessimistic, whilst if the spectrum is steep, \eqref{Random Diag++ bound} is much more optimistic and faster convergence is to be expected. The algorithm is robust to the spectrum of $A$, making the ``best of both worlds'' from projection and estimation. If the spectrum is steep, then we shall be able to capture most of the diagonal of $A$ through projection, whilst if it is flat, the stochastic estimation of the diagonal of the remainder $\widehat{A}$ will buffer the effects of projection. This is as is found in Figure \ref{fig:Synthetic Diag++}. Not only is Diag++ robust to spectra decaying at different rates, but also converges faster than simple stochastic estimation when the matrices in question have rapidly decaying eigenvalues. The algorithm even overtakes simple projection after a handful of vector queries in such cases. Thus, the answer to our original question at the start of this section is yes: we can improve stochastic estimation, particularly for what would otherwise be worst case scenarios.

\begin{figure}[h!]
    \centering
    \begin{tabular}{cc}
    \hspace{-1.2cm}
        \includegraphics[width=60mm]{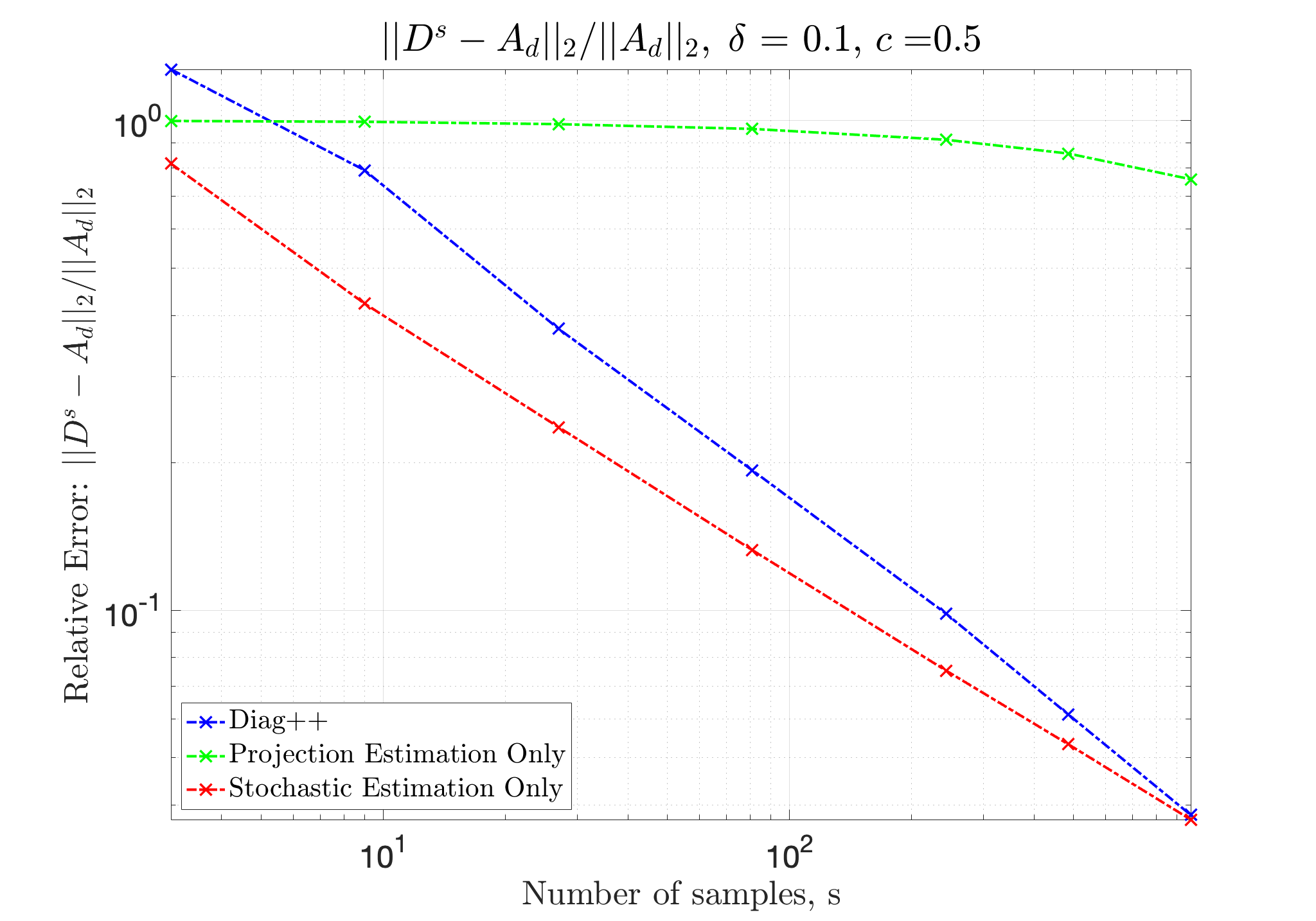} & \includegraphics[width=60mm]{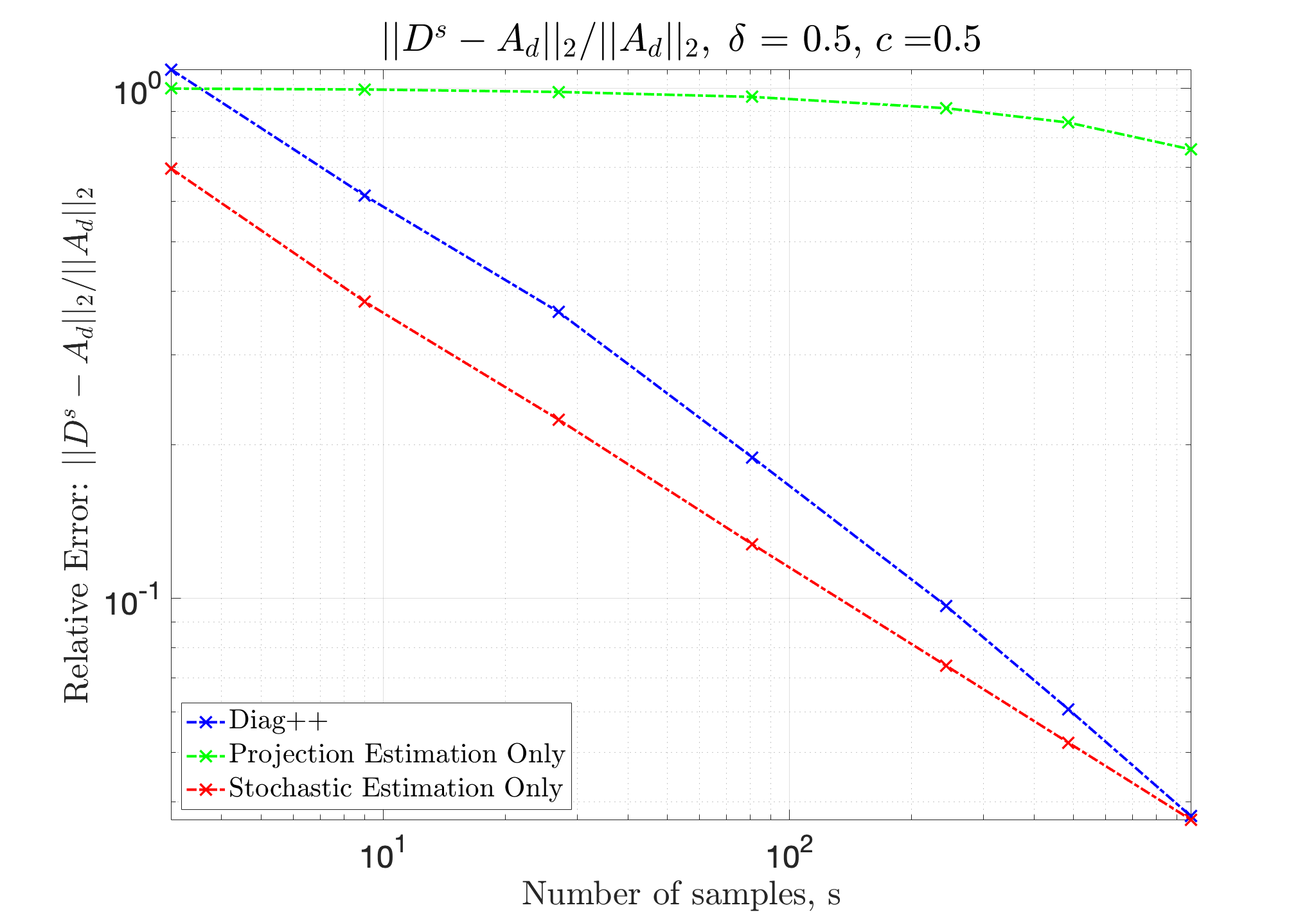} \\
    \hspace{-1.2cm}
        \includegraphics[width=60mm]{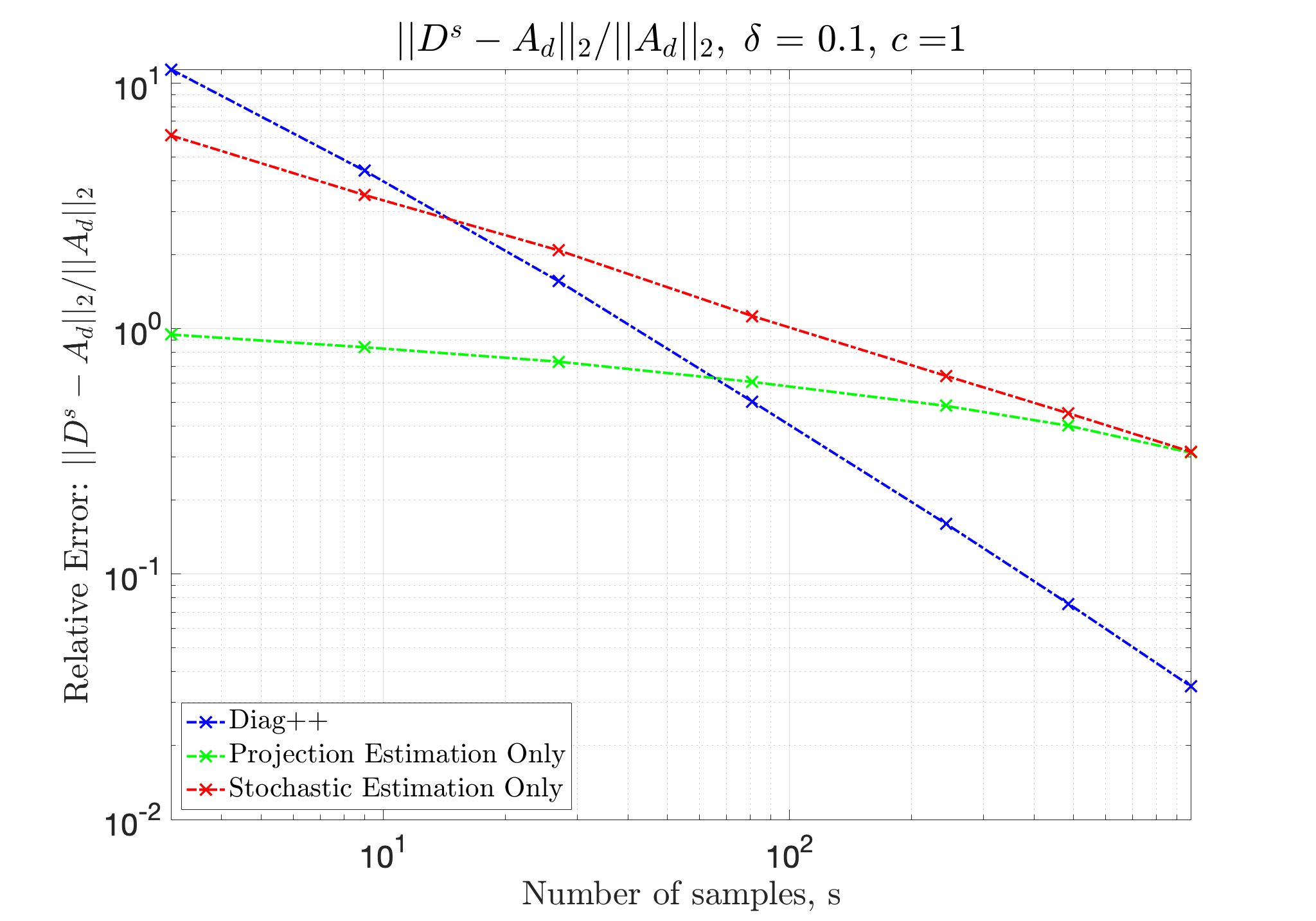} & \includegraphics[width=60mm]{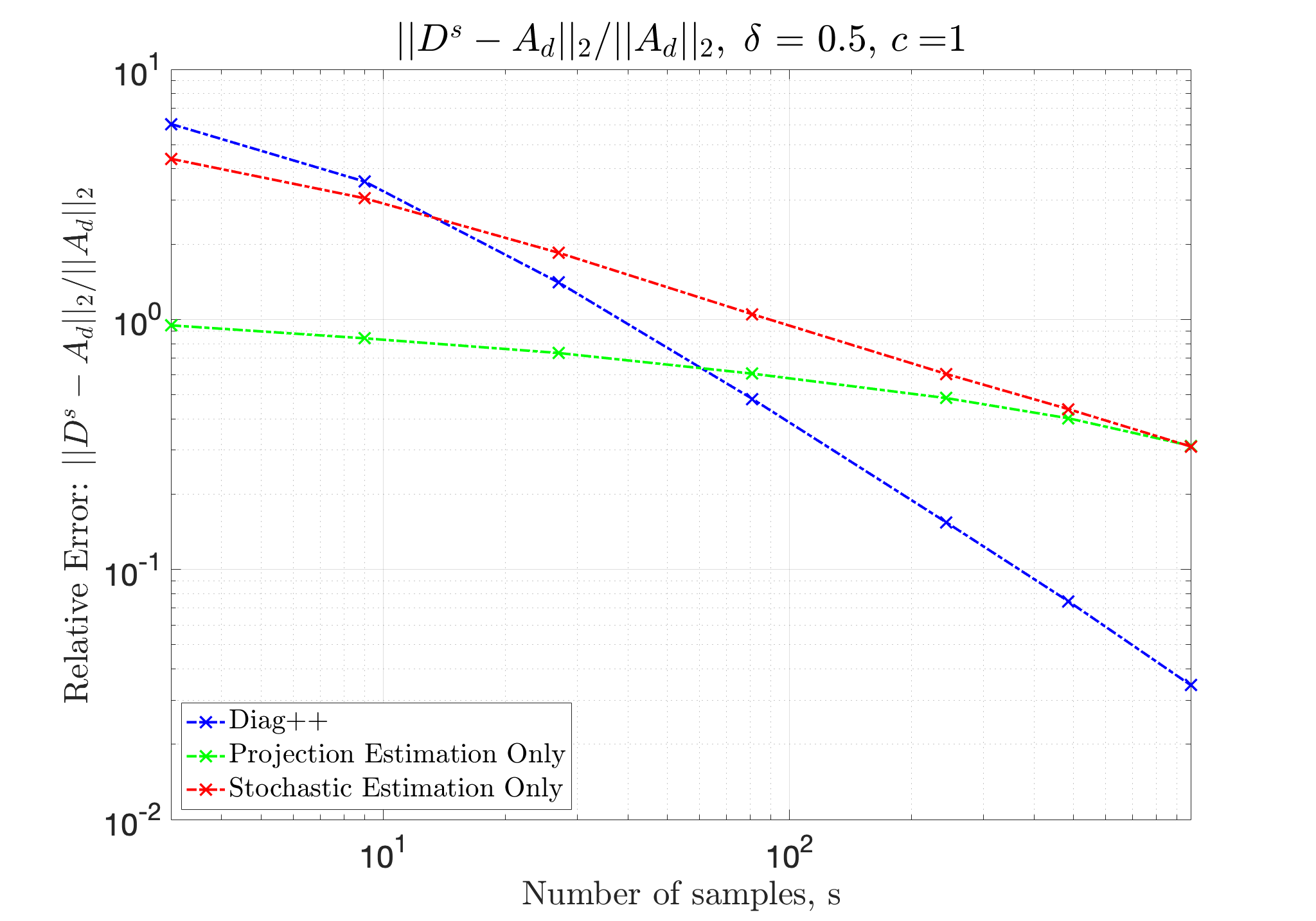}\\
    \hspace{-1.2cm}
        \includegraphics[width=60mm]{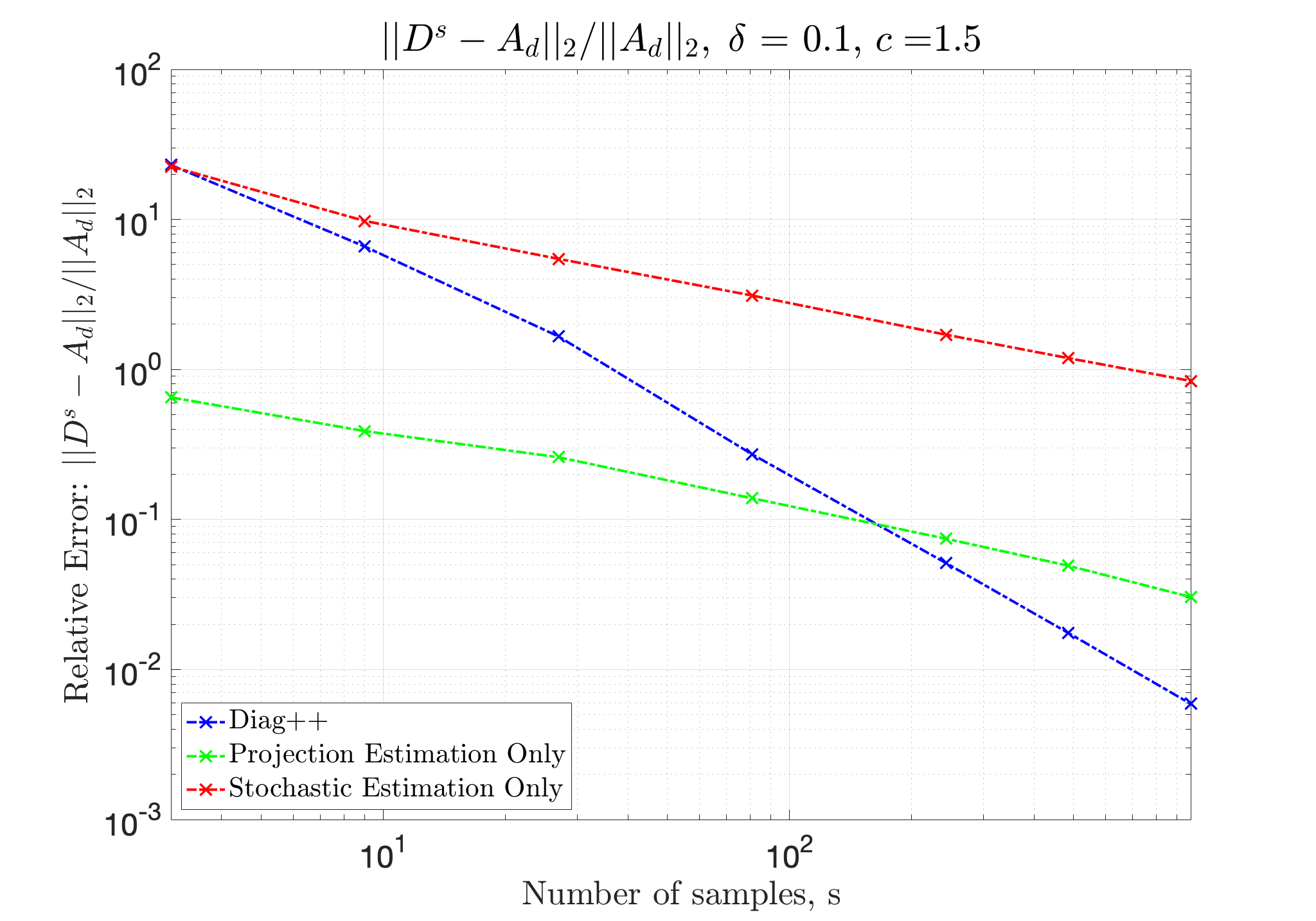} &
        \includegraphics[width=60mm]{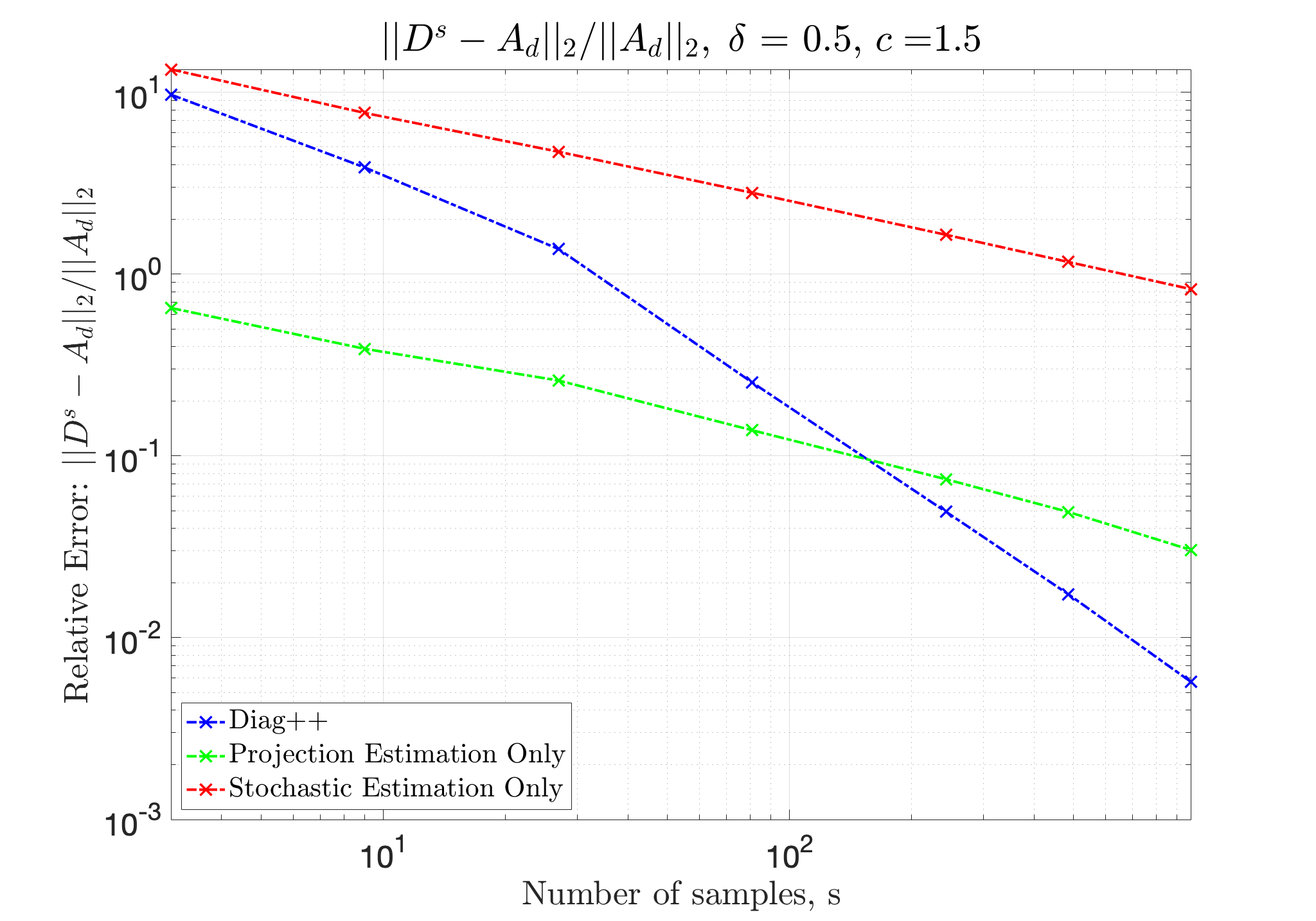}\\
    \end{tabular}
    \caption{\textit{Relative error versus number of matrix-vector queries: we report the median and $90^{th}$ percentile after 10 trials for Synthetic Matrices with $c = 0.5,\: 1,\: 1.5$. Whilst pure stochastic estimation outperforms Diag++ for slower eigenvalue decay, Diag++ converges at roughly the same rate. Estimation through projection only however, is poor. Meanwhile Diag++ obtains high performance for faster eigenvalue decay. It is clear that this algorithm is robust to the spectrum of the matrix in question, converging for both steep and flat eigenvalue distributions.}}
    \label{fig:Synthetic Diag++}
\end{figure}

%\bb{YN: referees might complain that the plots aren't distinguishable when printed in black-white. Let's replot if that happens}

\newpage
\section{Conclusion} \label{Conclusions}

\subsection{Summary}

In setting out, we sought to gain insight into particular questions surrounding diagonal estimation: 
\begin{enumerate}
    \item \label{Number 1} How many query vectors are needed to bound the error of the stochastic diagonal estimator in \eqref{First diagonal estimator}, and how does this change with distribution of vector entries?
    \item \label{Number 2} How does matrix structure affect estimation, and how might we improve worst case scenarios?

\end{enumerate}
In answer to \ref{Number 1}, we have seen that elemental error probabilistically converges as $O(1/\varepsilon^2)$, and scales as $\|\bm{A}_i\|_2^2 - A_{i i}^2$. We have seen that for the entire diagonal, similar results arise, with error going as $O(1/\varepsilon^2)$, and scaling as $\|A\|_F^2 - \|\bm{A}_d\|_2^2$. Additionally we have observed both theoretically and numerically, that if $s$ is small, error arising from the use of the Rademacher diagonal estimator is less than that arising from the Gaussian diagonal estimator. Conversely, if $s$ is large, there is an inconsequential difference between the two estimators.

 In answer to \ref{Number 2}, we have seen, as per our intuitions, that both eigenvector basis, and eigenvalue spectrum, influence the estimation of positive semi-definite matrices. Generally, the steeper the spectrum, the worse the stochastic diagonal estimator. However, we have introduced a new method of overcoming these worst case scenarios. We have analysed it to find favourable convergence properties, and demonstrated its superiority for estimating the diagonal of general dense matrices with steep spectra, whilst maintaining robustness to flatter spectra.\\

\subsection{Future Work}
There is ample scope to extend the work presented herein, both to applications, and further analysis. 
\begin{enumerate}
    \item We have seen that the error of a given diagonal element is dependent on the rest of its row. How can we ascertain the actual error of our estimates if we don't know the norm of the row of our implicit matrix? Or if we don't have an idea of its spectrum? Is the context of use of stochastic esimation a necessity? We note that we can make tracks on this question since we can estimate $\|A\|_F^2$ using trace estimation of $\tr(A^TA)$, a free byproduct of diagonal estimation.
    \item In our analysis we have perhaps wielded a hammer, rather than a scalpel, in order to arrive at the bounds for $\|\bm{D}_s - \bm{A}_d\|_2$. In particular, union bounding introduces the potentially unfavourable factor of $\ln(n)$. Perhaps future analysis could employ the dependence of element estimates in our favour to remove such a term?
    \item Having noted that Yao et.\ al \cite{yao2020adahessian} use only a single Rademacher estimator in their state-of-the-art second order machine learning optimiser, how much better will their method become with further estimation? Given the possibility of parallelised estimation, will this introduce a rapid new stochastic optimisation technique that outperforms current methods?
\end{enumerate}
These are interesting and exciting open questions, providing the potential investigator with both theoretical and numerical directions in which to head. During this paper, we also considered extending elemental diagonal analysis to shed new light on the original investigations into trace estimation. Whilst initial work suggested itself fruitful, the dependence of each element estimate upon one another confounded further analysis. We readily invite suggestions or insights into how these ideas may be useful in approaching trace estimation.

\bibliographystyle{abbrv}
\bibliography{bibliography.bib}   % name your BibTeX data base

\begin{thebibliography}{10}

\bibitem{adams2018estimating}
R.~P. Adams, J.~Pennington, M.~J. Johnson, J.~Smith, Y.~Ovadia, B.~Patton, and
  J.~Saunderson.
\newblock Estimating the spectral density of large implicit matrices.
\newblock {\em arXiv:1802.03451}, 2018.

\bibitem{ando1987singular}
T.~Ando, R.~A. Horn, and C.~R. Johnson.
\newblock The singular values of a {Hadamard} product: {A} basic inequality.
\newblock {\em Linear Multilinear Algebra}, 21(4):345--365, 1987.

\bibitem{avron2011randomized}
H.~Avron and S.~Toledo.
\newblock Randomized algorithms for estimating the trace of an implicit
  symmetric positive semi-definite matrix.
\newblock {\em Journal of the ACM}, 58(2):1--34, 2011.

\bibitem{baston2021thesis}
R.~A. Baston.
\newblock On matrix estimation.
\newblock Master's thesis, University of Oxford, 2021.

\bibitem{bekas2007estimator}
C.~Bekas, E.~Kokiopoulou, and Y.~Saad.
\newblock An estimator for the diagonal of a matrix.
\newblock {\em Appl. Numer. Math.}, 57(11-12):1214--1229, 2007.

\bibitem{boutsidis2017randomized}
C.~Boutsidis, P.~Drineas, P.~Kambadur, E.-M. Kontopoulou, and A.~Zouzias.
\newblock A randomized algorithm for approximating the log determinant of a
  symmetric positive definite matrix.
\newblock {\em Linear Algebra Appl.}, 533:95--117, 2017.

\bibitem{braverman2020schatten}
V.~Braverman, R.~Krauthgamer, A.~Krishnan, and R.~Sinoff.
\newblock Schatten norms in matrix streams: {Hello} sparsity, goodbye
  dimension.
\newblock In {\em ICML 2020}. PMLR, 2020.

\bibitem{chen2016accurately}
J.~Chen.
\newblock How accurately should {I} compute implicit matrix-vector products
  when applying the {Hutchinson} trace estimator?
\newblock {\em SIAM J. Sci. Comp}, 38(6):A3515--A3539, 2016.

\bibitem{cohen2018approximating}
D.~Cohen-Steiner, W.~Kong, C.~Sohler, and G.~Valiant.
\newblock Approximating the spectrum of a graph.
\newblock In {\em Proceedings of the 24th acm sigkdd international conference
  on knowledge discovery \& data mining}, pages 1263--1271, 2018.

\bibitem{cortinovis2021randomized}
A.~Cortinovis and D.~Kressner.
\newblock On randomized trace estimates for indefinite matrices with an
  application to determinants.
\newblock {\em Found. Comput. Math.}, pages 1--29, 2021.

\bibitem{davis2011university}
T.~A. Davis and Y.~Hu.
\newblock The {University} of {Florida} sparse matrix collection.
\newblock {\em ACM Transactions on Mathematical Software (TOMS)}, 38(1):1--25,
  2011.

\bibitem{di2016efficient}
E.~Di~Napoli, E.~Polizzi, and Y.~Saad.
\newblock Efficient estimation of eigenvalue counts in an interval.
\newblock {\em Numer. Lin. Alg. Appl.}, 23(4):674--692, 2016.

\bibitem{eckart1936approximation}
C.~Eckart and G.~Young.
\newblock The approximation of one matrix by another of lower rank.
\newblock {\em Psychometrika}, 1(3):211--218, 1936.

\bibitem{gambhir2017deflation}
A.~S. Gambhir, A.~Stathopoulos, and K.~Orginos.
\newblock Deflation as a method of variance reduction for estimating the trace
  of a matrix inverse.
\newblock {\em SIAM J. Sci. Comp}, 39(2):A532--A558, 2017.

\bibitem{goedecker1999linear}
S.~Goedecker.
\newblock Linear scaling electronic structure methods.
\newblock {\em Reviews of Modern Physics}, 71(4):1085, 1999.

\bibitem{goedecker1995tight}
S.~Goedecker and M.~Teter.
\newblock Tight-binding electronic-structure calculations and tight-binding
  molecular dynamics with localized orbitals.
\newblock {\em Physical Review B}, 51(15):9455, 1995.

\bibitem{golub1979generalized}
G.~H. Golub, M.~Heath, and G.~Wahba.
\newblock Generalized cross-validation as a method for choosing a good ridge
  parameter.
\newblock {\em Technometrics}, 21(2):215--223, 1979.

\bibitem{golub1997generalized}
G.~H. Golub and U.~Von~Matt.
\newblock Generalized cross-validation for large-scale problems.
\newblock {\em J. Comput. Graph. Stat.}, 6(1):1--34, 1997.

\bibitem{halko2011finding}
N.~Halko, P.-G. Martinsson, and J.~A. Tropp.
\newblock Finding structure with randomness: {Probabilistic} algorithms for
  constructing approximate matrix decompositions.
\newblock {\em SIAM Rev.}, 53(2):217--288, 2011.

\bibitem{han2017approximating}
I.~Han, D.~Malioutov, H.~Avron, and J.~Shin.
\newblock Approximating spectral sums of large-scale matrices using stochastic
  {Chebyshev} approximations.
\newblock {\em SIAM J. Sci. Comp}, 39(4):A1558--A1585, 2017.

\bibitem{han2015large}
I.~Han, D.~Malioutov, and J.~Shin.
\newblock Large-scale log-determinant computation through stochastic
  {Chebyshev} expansions.
\newblock In {\em ICML 2015}, pages 908--917. PMLR, 2015.

\bibitem{hutchinson1989stochastic}
M.~F. Hutchinson.
\newblock A stochastic estimator of the trace of the influence matrix for
  {Laplacian} smoothing splines.
\newblock {\em Commun. Stat. - Simul. Comput.}, 18(3):1059--1076, 1989.

\bibitem{li2020well}
J.~Li, A.~Sidford, K.~Tian, and H.~Zhang.
\newblock Well-conditioned methods for ill-conditioned systems: Linear
  regression with semi-random noise.
\newblock {\em arXiv:2008.01722}, 2020.

\bibitem{lin2017randomized}
L.~Lin.
\newblock Randomized estimation of spectral densities of large matrices made
  accurate.
\newblock {\em Numer. Math.}, 136(1):183--213, 2017.

\bibitem{lin2016approximating}
L.~Lin, Y.~Saad, and C.~Yang.
\newblock Approximating spectral densities of large matrices.
\newblock {\em SIAM Rev.}, 58(1):34--65, 2016.

\bibitem{meyer2021hutch++}
R.~A. Meyer, C.~Musco, C.~Musco, and D.~P. Woodruff.
\newblock Hutch++: {Optimal} stochastic trace estimation.
\newblock In {\em Symposium on Simplicity in Algorithms (SOSA)}, pages
  142--155. SIAM, 2021.

\bibitem{musco2020projection}
C.~Musco and C.~Musco.
\newblock Projection-cost-preserving sketches: {Proof} strategies and
  constructions.
\newblock {\em arXiv:2004.08434}, 2020.

\bibitem{musco2017spectrum}
C.~Musco, P.~Netrapalli, A.~Sidford, S.~Ubaru, and D.~P. Woodruff.
\newblock Spectrum approximation beyond fast matrix multiplication:
  {Algorithms} and hardness.
\newblock {\em arXiv:1704.04163}, 2017.

\bibitem{persson2021improved}
D.~Persson, A.~Cortinovis, and D.~Kressner.
\newblock Improved variants of the {H}utch++ algorithm for trace estimation.
\newblock {\em arXiv:2109.10659}, 2021.

\bibitem{roosta2015improved}
F.~Roosta-Khorasani and U.~Ascher.
\newblock Improved bounds on sample size for implicit matrix trace estimators.
\newblock {\em Foundations of Computational Mathematics}, 15(5):1187--1212,
  2015.

\bibitem{stathopoulos2013hierarchical}
A.~Stathopoulos, J.~Laeuchli, and K.~Orginos.
\newblock Hierarchical probing for estimating the trace of the matrix inverse
  on toroidal lattices.
\newblock {\em SIAM J. Sci. Comp}, 35(5):S299--S322, 2013.

\bibitem{tang2011domain}
J.~M. Tang and Y.~Saad.
\newblock Domain-decomposition-type methods for computing the diagonal of a
  matrix inverse.
\newblock {\em SIAM J. Sci. Comp}, 33(5):2823--2847, 2011.

\bibitem{ubaru2017applications}
S.~Ubaru and Y.~Saad.
\newblock Applications of trace estimation techniques.
\newblock In {\em International Conference on High Performance Computing in
  Science and Engineering}, pages 19--33. Springer, 2017.

\bibitem{wainwright2019high}
M.~J. Wainwright.
\newblock {\em High-dimensional statistics: A non-asymptotic viewpoint},
  volume~48.
\newblock Cambridge University Press, 2019.

\bibitem{yao2020video}
Z.~Yao and A.~Gholami.
\newblock With the authors: {AdaHessian} optimiser.
\newblock \url{https://youtu.be/S87ancnZ0MM?t=2000}.
\newblock Accessed: 24/08/2021.

\bibitem{yao2020adahessian}
Z.~Yao, A.~Gholami, S.~Shen, M.~Mustafa, K.~Keutzer, and M.~W. Mahoney.
\newblock {AdaHessian}: An adaptive second order optimizer for machine
  learning.
\newblock {\em arXiv:2006.00719}, 2020.

\end{thebibliography}

\end{document}